\documentclass[11pt]{article}
\usepackage{subfigure}
\usepackage{graphics,amsmath,amssymb,amsthm,accents}
\def \h#1{\widehat{#1}}
\def \t#1{\widetilde{#1}}

\newtheorem{lemma}{Lemma}
\newtheorem{theorem}{Theorem}
\newtheorem{proposition}{Proposition}
\numberwithin{equation}{section}
\numberwithin{lemma}{section}

\title{Solutions to the modified Korteweg-de Vries equation}

\author{Da-jun Zhang,\footnote{Corresponding author. E-mail: djzhang@staff.shu.edu.cn}
~~Song-lin Zhao,~~Ying-ying Sun,~~Jing Zhou\\
{\small\it Department of Mathematics,
 Shanghai University, Shanghai 200444,  P.R. China}}

\begin{document}
\maketitle
\date{}


%
%
%
%
%
%
%
%

\begin{abstract}
This is a continuation of Ref.\cite{ZDJ-arxiv}(arXiv:nlin.SI/0603008). In the present paper
we review solutions to the modified Korteweg-de Vries equation in
terms of Wronskians. The Wronskian entry vector needs to satisfy
a matrix differential equation set which contains complex operation.
This is different from the case of the Korteweg-de Vries equation.
We introduce an auxiliary matrix to deal with the complex operation
and then we are able to give complete solution expressions for the matrix differential
equation set. The obtained solutions to the modified Korteweg-de
Vries equation  can simply be categorized by  two types: solitons and
breathers, together with their limit cases. Besides, we give
rational solutions to the modified Korteweg-de Vries equation in
Wromskian form. This is derived with the help of the Galilean
transformed modified Korteweg-de Vries equation. Finally, typical
dynamics of the obtained solutions is analyzed and illustrated.
We list out the obtained solutions and their corresponding basic Wronskian vectors
in the conclusion part.

\vspace{0.5cm}
\noindent \textbf{Keywords}: {The modified Korteweg-de Vries equation, Wronskian,  breathers, rational solutions, dynamics}

\noindent \textbf{MSC 2010}: {37K40, 35Q35, 35Q35}

\noindent \textbf{PACS}: 02.30.Ik, 02.30.Jr, 05.45.Yv
\end{abstract}

\tableofcontents

\newpage

\section{Introduction}	

It is well known that the modified Kordeweg-de Vries (mKdV)
equation,
\begin{equation}
v_t+6\varepsilon v^2v_x + v_{xxx}=0,~~ \varepsilon=\pm 1,  \label{mKdV-pm}
\end{equation}
played an important role in constructing infinitely many conservation
laws\cite{MGK-1968} and Lax pair for the Korteweg-de Vries (KdV) equation.
The Lax pair led to the breakthrough of the Inverse Scattering Transform(IST)\cite{IST-1967}
and then the Soliton Theory.
The mKdV equation is also famous for its special soliton behavior, breathers.
Actually, the mKdV equation, or its Galilean transformed version,
\begin{equation}
{V_t}+ 12\varepsilon {v_0}V{V_X} + 6\varepsilon {V^2}{V_X} + {V_{XXX}} = 0,~~\varepsilon=\pm 1, \label{kdv-mkdv-pm}
\end{equation}
which is usually referred to as the mixed KdV-mKdV equation,
arose in many physics contexts, such as
anharmonic lattices\cite{
Soliton fission-72}, Alf\'ven waves\cite{Kono-Alven}, ion acoustic
solitons\cite{ion acoustic-40,Ion acoustic-69,Ion acoustic-79},
traffic jam\cite{jamming-74,jamming-70}, Schottky barrier transmission lines\cite{Schottky barrier-82},
thin ocean jets\cite{thin jets-73a,thin jets-73b}, internal waves\cite{internal solitary-36,Internal solitary-45},
heat pulses in solids\cite{heat pulses-63}, and so on.

With regard to exact solutions, many classical solving methods,
such as  Hirota's bilinear method\cite{Hirota-mKdV}, the IST\cite{Wadati-mKdV-IST,Tanaka-mKdV-IST},
commutation methods\cite{Gesztesy-Tams-1991}
and Wronskian technique\cite{Satsuma,Nimmo-Freeman-JPA,Gesztesy-Rmp-1991}
has been used to solve the mKdV equation. For more references
on  solutions of the mKdV equation, one can refer to Ref.\cite{Gesztesy-Tams-1991} and the references therein.

In general, for a soliton equation with bilinear form, its solutions
can be expressed through a Wronskian by imposing certain conditions on
its Wronskian entry
vector\cite{Freeman-Nimmo-KP,Sirianunpiboon-1988,Ma-You-KdV,ZDJ-arxiv}.
For convenience in the following we refer to such conditions as
condition equation set(CES). Usually for an (1+1)-dimensional soliton equation
the crucial  part in its CES is a coefficient matrix and the matrix
and its any similar form leads to same solutions for the corresponding
soliton equation. Thus it is possible to give a complete
classification (or structure) for the solutions of the
soliton equation by considering the canonical form of the
coefficient matrix\cite{Ma-You-KdV,ZDJ-arxiv}. It has been
understood that the solutions generated from a Jordan form
coefficient matrix are related to the solutions generated from a
diagonal form coefficient matrix via some limiting
proceure\cite{ZDJ-arxiv}. Therefore the latter solutions can be
referred to as limit solutions. Actually, from the viewpoint of the
IST, $N$ solitons are identified by $N$ distinct eigenvalues of the
corresponding spectral problem, or in other words, $N$ distinct
simple poles $\{k_j\}$ of transparent coefficient $\frac{1}{a(k)}$.
When $\{k_j\}$ are multiple-poles, the related multiple-pole
solution can be obtained through a limiting procedure like
$k_2\rightarrow k_1$ from simple-pole solution. This limiting
procedure is easily realized for solutions in Wronskian
form\cite{ZDJ-arxiv}. Such a procedure is also helpful to understand
the dynamics of limit solutions\cite{ZZZ-2009-PLA,ZDJ-TMP,KZS-sg}.

In Ref.\cite{ZDJ-arxiv} we mentioned four topics related to
solutions in Wronskian form: to find the CES, to solve the CES, to
describe relations between different kinds of solutions, and to
discuss dynamics of the solutions. In the present paper, following this line,
we will try to review the Wronskian solutions of the mKdV($\varepsilon=1$) equation
\begin{equation}
v_t+6v^2v_x + v_{xxx}=0, \label{mKdV}
\end{equation}
together with its Galilean transformed version,
\begin{equation}
{V_t}+ 12{v_0}V{V_X} + 6{V^2}{V_X} + {V_{XXX}} = 0.\label{kdv-mkdv}
\end{equation}

The main results of the paper are the following.
\begin{itemize}
\item{The CES for the mKdV equation can be given by
\begin{subequations}
\label{mKdV-CES}
\begin{align}
\varphi_{x}=&\mathbb{B}\bar{\varphi}, \label{mKdV-CES-b}\\
\varphi_{t}=&-4\varphi_{xxx},\label{mKdV-CES-c}
\end{align}
\end{subequations}
where $\varphi$ is the $N$-th order Wronskian entry vector, bar stands for complex conjugate
and $\mathbb{B}$ is a nontrivial $N\times N$ constant complex matrix.
We note that there is a complex conjugate involved in the CES
and this makes difficulties when solving  the CES.
}
\item{There are only two kinds of solutions led by the CES \eqref{mKdV-CES}:
Solitons (together with their limit case) and breathers (together with their limit case).
No rational solutions arise from \eqref{mKdV-CES} due to $|\mathbb{B}|\neq 0$.
}
\item{Rational solutions to the mKdV equation \eqref{mKdV}}
can be derived by means of the KdV-mKdV equation \eqref{kdv-mkdv}.
\item{Dynamics of obtained solutions is analyzed and illustrated.}
\end{itemize}

The paper is organized as follows. In Sec.2, we give a general CES
of the mKdV equation and simplify the CES by introducing  an auxiliary equation.
Then in Sec.3 we solve the CES and classify the solutions as solitons and breathers.
In Sec.4 we derive rational solutions to the mKdV equation. This is done with the help of
the KdV-mKdV equation \eqref{kdv-mkdv}.
Sec.5 consists of dynamic analysis and illustrations.
Finally, in the conclusion section we list out the obtained solutions and their corresponding basic Wronskian vectors.

\section{Wronskian solutions of the mKdV equation}

\subsection{Preliminary}
\label{sec:2.1}

An $N\times N$ Wronskian is defined as
\begin{equation}
W(\phi _{1},\phi _{2},\cdots ,\phi _{N})
=|\phi,\phi^{(1)},\cdots,\phi^{(N-1)}|=\left|
\begin{array}{cccc}
\phi _{1}^{(0)} & \phi _{1}^{(1)} & \cdots  & \phi _{1}^{(N-1)} \\
\phi _{2}^{(0)} & \phi _{2}^{(1)} & \cdots  & \phi _{2}^{(N-1)} \\
\vdots  & \vdots  &  \vdots & \vdots \\
\phi _{N}^{(0)} & \phi _{N}^{(1)} & \cdots  & \phi _{N}^{(N-1)}%
\end{array}
\right|,\label{Wro-def}
\end{equation}
where $\phi_{j}^{(l)}=\partial^l \phi_j/{\partial x}^l$ and
$\phi=(\phi_1, \phi_2,\cdots,\phi_N)^T$ is called the entry vector
of the Wronskian. Usually we use the compact
form\cite{Freeman-Nimmo-KP}
\begin{equation}
W(\phi)=|\phi,
\phi^{(1)},\cdots,\phi^{(N-1)}|=|0,1,\cdots,N-1|=|\widehat{N-1}|,
\label{wronskian}
\end{equation}
where $\widehat{N-j}$ indicates the set of consecutive
columns $0,1,\cdots,N-j$.
In the paper we also employ the notation $\t{N-j}$ to indicate the set of consecutive
columns $1,2,\cdots,N-j$.

A Wronskian can provide simple forms for its derivatives
and this advantage admits direct verification of solutions that are expressed in terms of Wronskians.
The following matrix properties are usually necessary in Wronskian verification.
\begin{proposition}\cite{Zhang-Hietarinta,ZDJ-arxiv}
\label{Prop 2.1}
Suppose that $\Xi$ is an
  $N\times N$ matrix with column vector set $\{\Xi_j\}$;  $\Omega$ is an
  $N\times N$ operator matrix with column vector set $\{\Omega_j\}$ and
  each entry $\Omega_{j,s}$ being an operator. Then
  we have
\begin{equation}
\sum^N_{j=1} |\Omega_j * \Xi|
=\sum^N_{j=1}|(\Omega^T)_{j} * \Xi^T|,
\label{eq-P-2.1}
\end{equation}
 where for any $N$-order column vectors $A_j$ and $B_j$ we define
\begin{equation}
A_j \circ B_j=(A_{1,j}B_{1,j},~A_{2,j}B_{2,j},\cdots, A_{N,j}B_{N,j})^T
\label{AjBj}
\end{equation}
 and
\begin{equation}
|A_j * \Xi|=|\Xi_1,\cdots,\Xi_{j-1},~A_j \circ\Xi_j,~\Xi_{j+1},\cdots, \Xi_{N}|.
\label{AjXi}
\end{equation}
\end{proposition}
\begin{proposition}\cite{Freeman-Nimmo-KP}
\label{Prop 2.2}
Suppose that $D$ is an $N\times(N-2)$ matrix and $a,~ b,~ c,~ d$ are
$N$th-order column vectors, then
\begin{equation}
|D,~a, ~b||D,~ c, ~d|-|D,~ a,~ c||D,~ b,~ d|+|D,~ a, ~d||D,~ b,~ c|=0.
\label{D-abcd}
\end{equation}
\end{proposition}

Besides, to write solutions of CES in  simple forms
one may make use of lower triangular Toeplitz matrices.
An $N$th-order lower triangular Toeplitz matrix means a matrix  in the following form
\begin{equation}
\mathcal{A}=\left(\begin{array}{cccccc}
a_0 & 0    & 0   & \cdots & 0   & 0 \\
a_1 & a_0  & 0   & \cdots & 0   & 0 \\
a_2 & a_1  & a_0 & \cdots & 0   & 0 \\
\cdots &\cdots &\cdots &\cdots &\cdots &\cdots \\
a_{N-1} & a_{N-2} & a_{N-3}  & \cdots &  a_1   & a_0
\end{array}\right)_{N\times N},~~~ a_j\in \mathbb{C}.
\label{A}
\end{equation}
All such matrices form a commutative semigroup $\widetilde{G}_N(\mathbb{C})$ with identity
with respect to matrix multiplication and inverse,
and the set $G_N(\mathbb{C})=\big \{\mathcal{A} \big |~\big. \mathcal{A}\in \widetilde{G}_N(\mathbb{C}),~|\mathcal{A}|\neq 0 \big\}$
makes an Abelian group.
Besides \eqref{A}, we will also need the following
block lower triangular Toeplitz matrix,
\begin{equation}
\mathcal{A}^B=\left(\begin{array}{cccccc}
A_0 & 0    & 0   & \cdots & 0   & 0 \\
A_1 & A_0  & 0   & \cdots & 0   & 0 \\
A_2 & A_1  & A_0 & \cdots & 0   & 0 \\
\cdots &\cdots &\cdots &\cdots &\cdots &\cdots \\
A_{N-1} & A_{N-2} & A_{N-3}  & \cdots &  A_1   & A_0
\end{array}\right)_{2N\times 2N},
\label{block-A-block}
\end{equation}
where
$A_j=\Big(\begin{array}{cc}
a_{j1} & 0   \\
0 & a_{j2}
\end{array}\Big)$
and $\{a_{js}\}$ are arbitrary complex numbers. All such block
matrices also compose a commutative semigroup with identity which we denote by
$\widetilde{G}^B_{2N}(\mathbb{C})$, and the set $G^B_{2N}(\mathbb{C})=\big
\{\mathcal{A}^B \big |~\big. \mathcal{A}^B\in
\widetilde{G}^B_N(\mathbb{C}),~|\mathcal{A}^B|\neq 0 \big\}$ makes
an Abelian group, too. If all the elements are real, then we
correspondingly denote the above mentioned matrix sets by
$\widetilde{G}_{N}(\mathbb{R})$, $G_{N}(\mathbb{R})$,
$\widetilde{G}^B_{2N}(\mathbb{R})$ and $G^B_{2N}(\mathbb{R})$. For
more properties of such matrices please refer to
Ref.\cite{ZDJ-arxiv}.

\subsection{CES of the mKdV equation}

By the transformation
\begin{equation}
v=i\Big(\ln{\frac{\bar{f}}{f}}\Big)_x=i \ \frac{ \bar{f}_{x}f
-\bar{f} f_x }{\bar{f}f}, \label{trans-mKdV}
\end{equation}
the mKdV equation \eqref{mKdV} can be bilinearized
as\cite{Hirota-mKdV}
\begin{subequations}
\begin{eqnarray}
&& (D_{t}+D_{x}^{3}) \bar{f} \cdot f=0,
\label{blinear-mKdV1}\\
&& D_{x}^{2}\bar{f} \cdot f=0, \label{blinear-mKdV2}
\end{eqnarray}
\label{blinear-mKdV}
\end{subequations}
where $i$ is the imaginary unit, $\bar{f}$ is the complex conjugate of $f$,
and $D$ is the well-known Hirota's bilinear operator defined by\cite{Hirota-1971,Hirota-book}
\begin{equation*}
D^m_{t}D^n_{x}a(t,x)\cdot b(t,x)
=\frac{\partial^m}{{\partial s}^m}\frac{\partial^n}{{\partial y}^n}
a(t+s,x+y)b(t-s,x-y)|_{s=0,y=0},~~m,n=0,1,\cdots.
\label{Hirota-ope}
\end{equation*}

The bilinear mKdV equation \eqref{blinear-mKdV} admits a solution $f$ in Wronskian form.
\begin{theorem}
\label{Th 2.1}
A Wronskian solution to the bilinear mKdV equation \eqref{blinear-mKdV} is given as
\begin{equation}
f=W(\phi)=|\widehat{N-1}|,
\label{wrons-mKdV}
\end{equation}
provided that its entry vector $\phi$ satisfies
\begin{subequations}
\begin{align}
\phi_{x}=& B(t)\bar{\phi},
\label{cond-a}\\
\phi_{t}=&-4\phi_{xxx}+C(t)\phi,
\label{cond-b}
\end{align}
\label{cond-mKdV}
\end{subequations}
where $B(t)=(B_{ij}(t))_{N\times N}$ and $C(t)=(C_{ij}(t))_{N\times
N}$ are two  $N\times N$  matrices of $t$ but independent of $x$,
and satisfy
\begin{subequations}
\begin{eqnarray}
&&|\bar{B}(t)|_t= 0,~~\mathrm{tr}{(C(t))}\in \mathbb{R}(t),\label{B-qiudao}\\
&&B_t(t)+B(t) \bar{C}(t)=C(t)B(t) \label{W_T}.
\end{eqnarray}
\label{compt}
\end{subequations}
\end{theorem}
The proof is given in \ref{A:sec-1}.

\subsection{Simplification of the CES \eqref{cond-mKdV}}
\label{sec:2.3}

To solve the CES \eqref{cond-mKdV} with arbitrary $B(t)$ and $C(t)$ which satisfy \eqref{compt}
we first introduce a non-singular $N\times N$  complex
matrix $H(t)\in \mathbb{C}_{N\times N}(t)$ such that (\cite{ODE-book}, also see \cite{ZDJ-arxiv})
\begin{equation}
H_t(t)=-H(t)C(t).
\end{equation}
By $H(t)$ we then introduce a new Wronskian entry vector
\begin{equation}
\psi=H(t)\phi,
\end{equation}
which transfers the CES \eqref{cond-mKdV} to the following,
\begin{subequations}
\begin{align}
\psi_{x}=& \widetilde{B}\bar{\psi},
\label{cond-a-s}\\
\psi_{t}=&-4\psi_{xxx},
\label{cond-b-s}
\end{align}
\label{cond-mKdV-s}
\end{subequations}
where $\widetilde{B}=H(t)B(t)\bar H^{-1}(t)$ has to be a constant matrix independent of both $x$ and $t$ due to
the compatibility condition \eqref{W_T} (noting that now $C(t)=0$).

We note that the Wronskians composed by $\phi$ and $\psi$, which we respectively denote by $f(\phi)$ and $f(\psi)$,
yield same solutions to the mKdV equation through the transformation \eqref{trans-mKdV}
due to $f(\psi)=|H(t)|f(\phi)$.
That means in the following one only needs to focus on the CES \eqref{cond-mKdV-s}.
However, since there exists a complex operation in \eqref{cond-mKdV-s},
solutions can not be classified in terms of the canonical form of $\widetilde{B}$, as done in \cite{Ma-You-KdV,ZDJ-arxiv}.
To overcome the difficulty we introduce an auxiliary  equation
\begin{equation}
\psi_{xx}= \widetilde{A}\psi,
\label{cond-a-s0}
\end{equation}
where $\widetilde{A}=\widetilde{B}\bar{\widetilde{B}}$.
Our plan now is   to first solve the equation set composed by \eqref{cond-a-s0} and \eqref{cond-b-s},
which is nothing but the CES of the KdV equation and have been well studied in Refs.\cite{Ma-You-KdV} and \cite{ZDJ-arxiv}.
Then in the second step we impose condition \eqref{cond-a-s} on the obtained solution $\psi$ and finally get solutions for \eqref{cond-mKdV-s}.

In addition, noting that  $\widetilde{A}$ and its any similar form $\mathbb{A}=P^{-1}\widetilde{A}P$
always generates same solutions to the mKdV equation, in the following let us focus on  the CES
\begin{subequations}
\label{mKdV-condition}
\begin{align}
\varphi_{xx}=&\mathbb{A}\varphi, \label{mKdV-condition-a}\\
\varphi_{x}=&\mathbb{B}\bar{\varphi}, \label{mKdV-condition-b}\\
\varphi_{t}=&-4\varphi_{xxx},\label{mKdV-condition-c}
\end{align}
\end{subequations}
where $\varphi=P^{-1}\psi$, $\mathbb{B}=P^{-1}\widetilde{B}\bar{P}$ and
\begin{equation}
\mathbb{A}=\mathbb{B}\bar{\mathbb{B}}. \label{cc}\\
\end{equation}
In the next section we will see that solutions to the mKdV equation can be classified in terms of the canonical form of $\mathbb{A}$
(rather than the canonical form of $\mathbb{B}$).

\section{Solutions of the mKdV equation}
In this section, we list several possible choices of $\mathbb{A}$
and derive the related Wronskian entry vectors,
one of which is for breathers.
We will also discuss the limiting relationship of some solutions.

First, for the matrix $\mathbb{A}$ defined by \eqref{cc}, we have the following result on its eigenvalues.
\begin{proposition}\label{Prop 3.1}
The eigenvalues of $\mathbb{A}$ defined by \eqref{cc} are either real or,
if there are some complex ones, appear as conjugate pairs.
\end{proposition}
We leave the proof in \ref{A:sec-2}.

According to the eigenvalues of $\mathbb{A}$, we can categorize solutions to the mKdV equation as solitons and breathers,
which correspond to real eigenvalues and complex eigenvalues of conjugate pairs, respectively.

\subsection{Solitons}\label{sec:3.1}

{\bf  Case I. Solitons:} ~When $\mathbb{A}$ has $N$ distinct real positive eigenvalues $\{\lambda^2_j\}$ its canonical form reads
\begin{equation}
\mathbb{A}={\rm Diag}( \lambda_{1}^2,  \lambda_{2}^2, ~ \cdots, ~
\lambda_{N}^2), \label{mathbb-A-11}
\end{equation}
where, for convenient to discuss,
we let  $\lambda_j=\varepsilon_j ||k_j||\neq 0$, in which $\varepsilon_j=\pm 1$ and $k_j$ can either real or complex numbers
with distinct absolute values.
We consider two subcases.

\noindent
(1).  $k_{j}\in \mathbb{R}$, i.e.
\begin{equation}
\mathbb{A}={\rm Diag}( k_{1}^2,  k_{2}^2, ~ \cdots, ~ k_{N}^2).
\end{equation}
Following the relation \eqref{cc}, we can take
\begin{equation}
\mathbb{B}={\rm Diag}( \pm k_{1},  \pm k_{2}, ~ \cdots, ~ \pm
k_{N}). \label{mathbb-B-11}
\end{equation}
 We neglect the sign $\pm$ because this can be compensated by the arbitrariness of $k_j$.
 So next we take
\begin{equation}
\mathbb{B}={\rm Diag}(  k_{1},   k_{2}, ~ \cdots, ~ k_{N}),~~ k_j\in \mathbb{R}.
\end{equation}
For the above matrix $\mathbb{B}$, the solution to the CES \eqref{mKdV-condition} can be
\begin{subequations}\label{gen-solitons}
\begin{equation}
\varphi=(\varphi_{1},  \varphi_{2}, \cdots,  \varphi_{N})^{T},
\label{gen-sol}
\end{equation}
where
\begin{equation}
\varphi^{}_{j}= a_{j}^+  e^{\xi_{j}}+ i a_{j}^-
 e^{-\xi_{j}}, ~\xi_{j}=k_{j}x-4k_{j}^{3}t+\xi_{j}^{(0)},~~
 a_{j}^+ , a_{j}^-, k_j, \xi_{j}^{(0)} \in \mathbb{R}.
\label{mkdv-c11}
\end{equation}
\end{subequations}
When $a^+_j=(-1)^{j-1},~a^-_j=1$, \eqref{mkdv-c11} reads
\begin{equation}
\varphi_j=(-1)^{j-1}e^{\xi_j}+i e^{-{\xi_j}}. \label{WH}
\end{equation}
In this case, the corresponding Wronskian solution \eqref{trans-mKdV} with  $f(\varphi)$
can be written as\cite{ZDJ-mkdvscs}
\begin{equation*}
f= \bigg (\prod^{N}_{j=1}e^{\xi_j}\bigg)
        \bigg(\prod_{1\leq j<l\leq N}(k_{j}-k_{l})\bigg)
    \sum_{\mu=0, 1}\exp \bigg\{\sum_{j=1}^{N}\mu_{j}(2\eta_{j}+\frac{\pi}{2}i)+\sum_{1\leq j<l\leq N}
   \mu_{j}\mu_{l}A_{jl} \bigg\},
\end{equation*}
where the sum over $\mu=0,1$ refers to each of $\mu_j=0,1$ for $j=1,2,\cdots, N$, and
\begin{equation*}
\eta_{j}=-\xi_{j}-\frac{1}{4}\sum_{l=2,l\neq j}^{N}A_{jl},
~~e^{A_{jl}}=\bigg(\frac{k_l-k_j}{k_l+k_j}\bigg)^2.
\end{equation*}
This coincides with the $N$-soliton solution in Hirota's exponential polynomial form\cite{Ablowitz-book}.

\vskip 5pt
\noindent
(2).  $k_j=k_{j1}+ik_{j 2} \in \mathbb{C}$, i.e.
\begin{equation}
\mathbb{A}={\rm Diag}( \lambda_{1}^2,  \lambda_{2}^2, ~ \cdots, ~
\lambda_{N}^2),  ~~\lambda_{j}^2=k_{j1}^2+k_{j2}^2. \label{mathbb-A-12}
\end{equation}
In this case we have
\begin{equation}
\mathbb{B}={\rm Diag}(   k_{1},    k_{2}, ~ \cdots, ~
  k_{N}),~~k_j \in \mathbb{C}, \label{mathbb-B-12}
\end{equation}
and the  solution to the CES \eqref{mKdV-condition} can be given by \eqref{gen-sol} with
\begin{subequations}
\begin{equation}
 \varphi_j= \gamma_j(a_{j}^+  e^{\xi_{j}}+ i a_{j}^-
e^{-\xi_{j}}), ~\xi_{j}=\lambda_{j}x-4\lambda_{j}^{3}t+\xi_{j}^{(0)}, ~~
a_{j}^+ , a_{j}^-,~ \xi_{j}^{(0)} \in \mathbb{R},
\label{mkdv-c12}
\end{equation}
where
\begin{eqnarray}
\gamma_j=1+\frac{i(\lambda_j-k_{j1})}{k_{j2}},~~
\lambda_j=\varepsilon_j ||k_j||.
\end{eqnarray}
\end{subequations}
Now, if we compare \eqref{mkdv-c11} and \eqref{mkdv-c12}, we can find that
both of them provide same solution to the mKdV equation
through the transformation \eqref{trans-mKdV}.
In particular, when $k_{j1}=0$, i.e.,
\begin{equation}
\mathbb{B}={\rm Diag}(  i k_{12},   i k_{22}, ~ \cdots, ~ i k_{N2}),
\label{B-con-im}
\end{equation}
and we take $\lambda_j=k_{j2}$,
the entry function \eqref{mkdv-c12} reduces to
\begin{equation}
\varphi_{j}=(1+ i) (a_{j}^+  e^{\xi_{j}}+ i a_{j}^-
e^{-\xi_{j}}),~\xi_{j}=k_{j}x-4k_{j}^{3}t+\xi_{j}^{(0)},~ a_{j}^+ , a_{j}^-, \xi_{j}^{(0)} \in \mathbb{R}.
\label{mkdv-c13}
\end{equation}
This was first given by Nimmo and Freeman\cite{Nimmo-Freeman-JPA} as Wronskian entries for soliton solutions.

Let us remark this case as follows.

\vskip 5pt
\noindent
\textbf{Remarks}:
\begin{itemize}
\item{
\textit{
When $\mathbb{B}$ in the CES \eqref{mKdV-condition-b} and \eqref{mKdV-condition-c}
is a diagonal matrix \eqref{mathbb-B-12}
of which the diagonal elements have different absolute values,
no matter these diagonal elements are real or complex, the related Wronskian $f(\varphi)$
generates $N$-soliton solutions to the mKdV equation}.}
\item{\textit{Besides, we specify the non-degenerate condition\footnote{
 The above non-degenerate relation can also be described as follows.
Define the equivalent relation $\sim$ on the complex plane $\mathbb{C}$:
\[ k_i\sim k_j~~\mathrm{iff}~~||k_i||= ||k_j||.\]
The quotient space $\mathbb{C}/\sim$ denotes the positive half real axis.
Then to get a non-degenerate soliton solution one needs
$k_i\nsim k_j,~(i\neq j)$.}
\begin{equation}
||k_i||\neq ||k_j||,~~(i\neq j),
\end{equation}
i.e., $k_i$ and $k_j$ $(i \neq j)$ can not appear on the same circle
with the original point as the center of the circle.
Otherwise, the solution degenerates.
}
}
\end{itemize}

\vskip 10pt
\noindent
{\bf  Case II. Limit solutions of solitons:} Corresponding to Case I, we discuss two subcases.

\noindent
(1). Let
\begin{equation}
\mathbb{A}=\left(\begin{array}{ccccccc}
k^2_1 & 0    & 0   & \cdots & 0   & 0 & 0 \\
2k_1   & k^2_1  & 0   & \cdots & 0   & 0 & 0\\
1   & 2k_1    & k^2_1 & \cdots & 0   & 0 & 0\\
\cdots &\cdots &\cdots &\cdots &\cdots &\cdots &\cdots \\
0   & 0    & 0   & \cdots & 1 & 2k_1  & k^2_1
\end{array}\right)_N, ~~ k_1\in \mathbb{R}.
\label{mathbb-A-21}
\end{equation}
In this subcase, general solution to the equation set  \eqref{mKdV-condition-a} and \eqref{mKdV-condition-c} is\cite{ZDJ-arxiv}
\begin{equation}
\hat{\phi}=\mathcal{A}\mathcal{Q}^{+}_{0}+ \mathcal{B} \mathcal{Q}^{-}_{0},~~
\mathcal{A}, \mathcal{B} \in \widetilde{G}_N(\mathbb{C}), \\
\label{gen-sol-L-1}
\end{equation}
where
\begin{equation}
\label{mathbb-Q}
 \mathcal{Q}^{\pm}_0=(\mathcal{Q}^{\pm}_{0, 0},
\mathcal{Q}^{\pm}_{0,1}, \cdots, \mathcal{Q}^{\pm}_{0, N-1})^T,~~
\mathcal{Q}^{\pm}_{0, s}=\frac{1}{s!}\partial^{s}_{k_1}e^{\pm
\xi_1},
\end{equation}
$\xi_1$ is defined in \eqref{mkdv-c11},
$\widetilde{G}_N(\mathbb{C})$ is the commutative set of all the
$N$th-order lower triangular Toeplitz matrices, see
Sec.\ref{sec:2.1}.

For the matrix $\mathbb{B}$ satisfying \eqref{cc} we take
\begin{equation}
\mathbb{B}=\left(\begin{array}{cccccc}
k_1 & 0    & 0   & \cdots & 0   & 0 \\
1   & k_1  & 0   & \cdots & 0   & 0 \\
\cdots &\cdots &\cdots &\cdots &\cdots &\cdots \\
0   & 0    & 0   & \cdots & 1   & k_1
\end{array}\right)_N.
\label{mathbb-B-21}
\end{equation}
Then, substituting  into \eqref{mKdV-condition-b} yields
\begin{equation}
\mathcal{A} \mathcal{Q}^{+}_{0, x}+\mathcal{B} \mathcal{Q}^{-}_{0,x}=\mathbb{B}(\bar{\mathcal{A}} \mathcal{Q}_0^{+}
+\bar{\mathcal{B}}\mathcal{Q}_0^{-}).
\label{gen-sol-relate}
\end{equation}
Meanwhile, it can be verified that
\begin{equation}
\mathcal{Q}^{+}_{0, x}=\mathbb{B} \mathcal{Q}_{0}^{+},~~
 \mathcal{Q}^{-}_{0, x}=-\mathbb{B}\mathcal{Q}_{0}^{-}.
\label{gen-sol-relate-1}
\end{equation}
Noting that $\mathbb{B}\in \widetilde{G}_N(\mathbb{R})$,
substituting \eqref{gen-sol-relate-1} into \eqref{gen-sol-relate},
and making use of the commutative property of
$\widetilde{G}_N(\mathbb{C})$, we have
\begin{equation}
\mathbb{B}(\mathcal{A} \mathcal{Q}^{+}_{0}-\mathcal{B}
\mathcal{Q}^{-}_{0})=\mathbb{B}(\bar{\mathcal{A}} \mathcal{Q}_0^{+}
+\bar{\mathcal{B}}\mathcal{Q}_0^{-}). \label{gen-sol-relate-new}
\end{equation}
Then, compared with \eqref{gen-sol-relate} we immediately get
\begin{eqnarray}
 \mathcal{A}=\bar {\mathcal{A}},~~
\mathcal{B}=-\bar {\mathcal{B}},
\label{gen-sol-related-b}
\end{eqnarray}
which means $\mathcal{A}$ is real and $\mathbb{B}$ pure imaginary.
In the end, the solution to the CES \eqref{mKdV-condition} can be described as
\begin{equation}
{\varphi}=\mathcal{A^+} \mathcal{Q}_{0}^{+}+ i \mathcal{A^-}
\mathcal{Q}_{0}^{-},~~ \mathcal{A^\pm}\in
\widetilde{G}_N(\mathbb{R}). \label{gen-sol-21}
\end{equation}

We note that solutions generated from \eqref{gen-sol-21} can also be derived from
the solution given in Case I.(1) by a limiting procedure(cf.\cite{ZDJ-arxiv}).
Let us explain this procedure by starting from the following Wronskian
\begin{eqnarray}
f(\varphi)=\frac{W(\varphi_1,\varphi_2,\cdots,\varphi_{N})}{\prod_{j=2}^{N}(k_j-k_1)^{j-1}}
\label{solution-1.1-limit}
\end{eqnarray}
with $\varphi_1=\varphi_1(k_1,x,t)=a_1^+e^{\xi_1}+ia_1^-e^{-\xi_1}$ as defined in
\eqref{mkdv-c11} and $\varphi_j=\varphi_1(k_j,x,t)$.
\eqref{solution-1.1-limit} gives a Wronskian solution to the
bilinear mKdV equation \eqref{blinear-mKdV}. Taking the limit $k_j
\rightarrow k_1$ successively for $j=2,3,\cdots,N$ and using
L'Hospital rule, the Wronskian \eqref{solution-1.1-limit} goes to a
Wronskian $W(\varphi)$ with entry vector \eqref{gen-sol-21}, where
the arbitrary coefficient matrices $\mathcal{A}, \mathcal{B}$ can
come from $a_1^{\pm}$ by considering $a_1^{\pm}$ to be some
polynomials of $k_1$ (cf.\cite{ZDJ-arxiv}).

\vskip 5pt
\noindent
(2). Corresponding to Case I.(2), let us consider
\begin{equation}
\mathbb{A}=\left(\begin{array}{ccccccc}
k_{11}^2+k_{12}^2 & 0    & 0   & \cdots & 0   & 0 & 0 \\
2k_{12}   & k_{11}^2+k_{12}^2  & 0   & \cdots & 0   & 0 & 0\\
1   & 2k_{12}    & k_{11}^2+k_{12}^2& \cdots & 0   & 0 & 0\\
\cdots &\cdots &\cdots &\cdots &\cdots &\cdots &\cdots \\
0   & 0    & 0   & \cdots & 1 & 2k_{12}  & k_{11}^2+k_{12}^2
\label{A-II-2}
\end{array}\right)_N,
\end{equation}
where $k_{11}, k_{12} \in \mathbb{R}$ and $k_{12}\neq 0$.
It then follows from \eqref{cc} that one can take
\begin{equation}
\mathbb{B}=\left(\begin{array}{cccccc}
-ik_1 & 0    & 0   & \cdots & 0   & 0 \\
1   & -ik_1  & 0   & \cdots & 0   & 0 \\
\cdots &\cdots &\cdots &\cdots &\cdots &\cdots \\
0   & 0    & 0   & \cdots & 1   & -ik_1
\end{array}\right)_N,~~
k_1=k_{11}+ik_{12}.
\label{mathbb-B-22}
\end{equation}

For the matrix $\mathbb{A}$ defined  by \eqref{A-II-2}, the general  solution
to the equation set \eqref{mKdV-condition-a} and \eqref{mKdV-condition-c} can be written as
\begin{equation}
{\varphi}=\mathcal{A}\mathcal{P}^{+}_{0}+ \mathcal{B}
\mathcal{P}^{-}_{0},~~
 \mathcal{A}, \mathcal{B} \in \widetilde{G}_N(\mathbb{C}),
\label{gen-sol-22}
\end{equation}
where\footnote{Here we define $\mathcal{P}^{\pm}_{0, s}$ by
taking derivative with respect to $k_{12}$ rather than $k_{11}$ because in this case we always have $k_{12}\neq 0$.}
\begin{subequations}
\begin{eqnarray}
\mathcal{P}^{\pm}_0&& =(\mathcal{P}^{\pm}_{0, 0},
\mathcal{P}^{\pm}_{0, 1}, \cdots, \mathcal{P}^{\pm}_{0, N-1})^T, ~
 \mathcal{P}^{\pm}_{0, s}=\frac{1}{s!}\partial^{s}_{k_{12}}e^{\pm \xi_1}, \\
 \xi_1&&=\lambda_1 x-4\lambda^3_1 t+\xi_1^{(0)}, ~
\lambda_1=\sqrt{k_{11}^2+k_{12}^2}.
\end{eqnarray}
\label{gen-sol-22-xi}
\end{subequations}
Next, we substitute \eqref{gen-sol-22} together with \eqref{gen-sol-22-xi} into
\eqref{mKdV-condition-b} so that we identify $\mathcal{A}, \mathcal{B}$ for \eqref{mKdV-condition-b}.
This substitution yields
\begin{equation}
\mathcal{A} \mathcal{P}^{+}_{0, x}+\mathcal{B} \mathcal{P}^{-}_{0, x}\\
=\mathbb{B}(\bar{\mathcal{A}} \mathcal{P}_0^{+}
+\bar{\mathcal{B}}\mathcal{P}_0^{-})\\
=\bar{\mathcal{A}}\mathbb{B}\mathcal{P}^{+}_0+\bar{\mathcal{B}}\mathbb{B}\mathcal{P}^{-}_0,
\label{gen-sol-relate4q"}
\end{equation}
where we have made use of the commutative property of
$\widetilde{G}_N(\mathbb{C})$. Then, noting that
\begin{equation}
\mathcal{P}_{0, x}^{+} =W \mathcal{P}_{0}^{+},~\mathcal{P}_{0,
x}^{-} = - W \mathcal{P}_{0}^{-}, \label{W-rela}
\end{equation}
where
\begin{equation}
W =\left(\begin{array}{ccccccc}
w_0 & 0    & 0   & \cdots    & 0 & 0 \\
w_1  & w_0  & 0   & \cdots    & 0 & 0 \\
w_2 & w_1    & w_0 & \cdots    & 0 & 0\\
\cdots &\cdots &\cdots &\cdots &\cdots &\cdots  \\
w_{N-1} &w_{N-2} &w_{N-3}& \cdots & w_1 & w_0
\end{array} \right)_N ,~w_{j}=\frac{1}{j!}\partial^{j}_{k_{12}} \lambda_1,
\label{W}
\end{equation}
and substituting \eqref{W-rela} into \eqref{gen-sol-relate4q"}, we get the relation
\begin{equation}
 \mathcal{A} W = \bar {\mathcal{A}}\mathbb{B},~
-\mathcal{B}W =\bar {\mathcal{B}} \mathbb{B}. \label{AB-W}
\end{equation}
To have a clearer result, we write
\begin{equation}
\mathcal{A}=\mathcal{A}_1+i \mathcal{A}_2, ~
\mathcal{B}=\mathcal{B}_1+i\mathcal{B}_2,~
\mathbb{B}=\mathbb{B}_{1}+i\mathbb{B}_{2},
\end{equation}
where $\mathcal{A}_1,\mathcal{A}_2,\mathcal{B}_1,\mathcal{B}_2$ are
in $\widetilde{G}_N(\mathbb{R})$ and
\begin{equation}
\mathbb{B}_{1} =\left(\begin{array}{cccc} k_{12}&~~& ~~ & 0
\\1&k_{12}&~~ & ~~\\~~&\ddots&\ddots &~~\\0&~~&1&k_{12} \end{array}\right), ~~
\mathbb{B}_{2} =\left(\begin{array}{cccc} -k_{11}&~~& ~~ & 0
\\0& -k_{11}&~~ & ~~\\~~&\ddots&\ddots &~~\\0&~~&0& -k_{11} \end{array}\right).
\end{equation}
Then \eqref{AB-W} yields
\begin{equation}
\label{tri-relation-A} \mathcal {A}_1(W-\mathbb{B}_1 )-\mathcal
{A}_2\mathbb{B}_2 =0,~
       \mathcal {A}_1 \mathbb{B}_2 -\mathcal {A}_2 (\mathbb{B}_1 +W)=0,
\end{equation}
\begin{equation}
\label{tri-relation-B} \mathcal {B}_1(W+\mathbb{B}_1 )+\mathcal
{B}_2 \mathbb{B}_2 =0,   ~
       \mathcal {B}_1 \mathbb{B}_2 -\mathcal {B}_2 (\mathbb{B}_1 -W)=0,
\end{equation}
which provides the equation sets for $\mathcal{A}_1, \mathcal{A}_2,
\mathcal{B}_1, \mathcal{B}_2$. We note that since all the elements
in \eqref{tri-relation-A} and \eqref{tri-relation-B} are in
$\widetilde{G}_N(\mathbb{R})$, we can treat \eqref{tri-relation-A}
and \eqref{tri-relation-B} as ordinary linear equation sets. Let us
look at \eqref{tri-relation-A}. It first indicates
\begin{equation}
\mathcal {A}_2=\mathbb{B}^{-1}_{2}(W-\mathbb{B}_{1}) \mathcal {A}_1,
\end{equation}
and further, by eliminating $\mathcal {A}_2$,
\begin{equation}
(\mathbb{B}^2_{1}+\mathbb{B}^{2}_{2}-W^2)\mathcal{A}_1=0.
\label{abw}
\end{equation}
We note that
\begin{equation}
\mathbb{B}^2_{1}+\mathbb{B}^{2}_{2}-W^2=0.
\label{BBW=0}
\end{equation}
In fact, by calculation we find
\[(W^2)_{i,j}=\left\{
\begin{array}{ll}
0,& i<j,\\
\frac{1}{(i-j)!}\partial^{i-j}_{k_{12}}(k_{11}^2+k_{12}^2), ~& i\geq j,
\end{array}\right.
\]
from which it is easy to verify \eqref{BBW=0}. Thus, it turns out that
$\mathcal{A}_1$ can be an arbitrary element in
$\widetilde{G}_N(\mathbb{R})$ and then
\begin{equation}
\mathcal {A}=\mathcal {A}_1\mathcal {M}, ~
\mathcal{M}=I_N+i\mathbb{B}^{-1}_{2}(W-\mathbb{B}_{1}),
\label{M}
\end{equation}
where $I_N$ is the $N$th-order unit matrix.
Similarly, we can find
\[
\mathcal {B}=i\mathcal{B}_2\mathcal {M},~\mathcal{B}_2\in
\widetilde{G}_N(\mathbb{R}).
\]

Thus we can conclude that the solution to the CES \eqref{cond-mKdV}
can be given by
\begin{equation}
\varphi= \mathcal {M}(\mathcal{A}^+\mathcal{P}^{+}_{0}+i~\mathcal
{A}^{-}\mathcal{P}^{-}_{0}),~\mathcal{A}^{\pm}\in
\widetilde{G}_N(\mathbb{R}),
\end{equation}
where $\mathcal{M}$ is defined in \eqref{M}. Obviously, the matrix
$\mathcal{M}$ contributes nothing through the transformation
\eqref{trans-mKdV} to the mKdV equation. Therefore in practice we
may remove $\mathcal{M}$ and use the effective part
\begin{equation}
\varphi= \mathcal{A}^+\mathcal{P}^{+}_{0}+i~\mathcal
{A}^{-}\mathcal{P}^{-}_{0},~\mathcal{A}^{\pm}\in
\widetilde{G}_N(\mathbb{R}),
\end{equation}
which is as same as \eqref{gen-sol-21}.

\subsection{Breathers}
{\bf  Case III. Breathers}: When $\mathbb{A}$ has distinct complex (conjugate-pair) eigenvalues, we may have breather solutions.
Let us consider a $2N$-th order matrix,
\begin{equation}
\mathbb{A}={\rm Diag}( k_{1}^2,  \bar{k}_{1}^2, ~ \cdots, ~
k_{N}^2, \bar{k}_{N}^2)_{2N},~k_{j}\neq 0,
\label{mKdV-matrix-C4}
\end{equation}
where $k_j\in \mathbb{C},~ j=1,2,\cdots,N$. The matrix $\mathbb{B}$ which generates breathers is
\begin{equation}
\mathbb{B}={\rm Diag}(\Theta_{1},  \Theta_{2}, \cdots,
\Theta_{N})_{2N}, ~~ \Theta_{j}=\left(\begin{array}{cc}
0 & k_{j}  \\
\bar{k}_{j} & 0  \\
\end{array}\right).
\end{equation}
For this case, the Wronskian entry vector, i.e., the solution to the CES \eqref{mKdV-condition}, can be taken as
\begin{subequations}
\label{breather}
\begin{equation}
{\varphi}=(\varphi_{11}, \varphi_{12}, \varphi_{21}, \varphi_{22},
\cdots, \varphi_{N1}, \varphi_{N2})^{T},
\end{equation}
where
\begin{eqnarray}
\varphi_{j1}&&= a_{j}  e^{\xi_j}+b_{j} e^{-{\xi}_j}, ~
\varphi_{j2}=\bar{a}_{j}  e^{\bar{\xi}_j} -\bar{b}_{j}
e^{-{\bar{\xi}_j}}, \nonumber\\
\xi_j&& =k_{j}x- 4 k_j^3 t+ \xi_{j}^{(0)},~a_{j}, b_{j},\xi_{j}^{(0)}  \in \mathbb{C}.
\label{xi-breather}
\end{eqnarray}
\end{subequations}
With the above $\varphi$ as basic entry vector, the Wronskian $f(\varphi)$ will provide breather solutions for the mKdV equation.
For non-trivial solutions we need $k_{j1}\neq 0$, (see Sec.\ref{sec:5.3}).

\vskip 10pt
\noindent
{\bf  Case IV. Limit solutions of breathers}:
In this case let us consider the following block matrix,
\begin{subequations}
 \begin{equation}
\mathbb{A}=\left(
\begin{array}{ccccccc}
  \mathcal {K} & 0 & 0 & \ldots & 0 & 0 & 0\\
 \widetilde{\mathcal {K}} &\mathcal {K} & 0 & \ldots & 0 & 0 & 0\\
 I_2& \widetilde{\mathcal {K}} &\mathcal {K}  & \ldots & 0 & 0 & 0\\
   \ldots &  \ldots &  \ldots &  \ldots &  \ldots &  \ldots&  \ldots\\
 0 & 0 & 0 & \ldots & I_2& \widetilde{\mathcal {K}} &\mathcal {K}
\end{array}\right)_{2N},
\end{equation}
where
\begin{equation}
  \mathcal {K} = \left(
\begin{array}{cc}
k^2_1&  0 \\
 0 & \bar{k}_1^2  \\
\end{array}\right), ~~
\widetilde{\mathcal {K}} = \left(
\begin{array}{cc}
2k_1 &  0  \\
 0 & 2\bar{k}_1  \\
\end{array}\right),~~I_2=\left(
\begin{array}{cc}
1 &  0  \\
   0 & 1  \\
\end{array}\right).
\end{equation}
\end{subequations}
The matrix $\mathbb{B}$ satisfying \eqref{cc} can
be taken as
\begin{subequations}
\begin{equation}
 \mathbb{B}=\left(
\begin{array}{ccccc}
  \widetilde{B} & 0 & \ldots & 0 & 0 \\
  \widetilde{I}_2 & \widetilde{B} & \ldots & 0 & 0 \\
   \ldots &  \ldots &  \ldots &  \ldots &  \ldots\\
  0 &  0 & \ldots &  \widetilde{I}_{2} & \widetilde{B}
\end{array}\right)_{2N},
\end{equation}
where
\begin{equation}
\widetilde{I}_{2} = \left(
\begin{array}{cc}
0 &  1  \\
 1 & 0  \\
\end{array}\right), ~~\widetilde{B}=\left(
\begin{array}{cc}
0 &  k_{1}  \\
  \bar{ k}_{1} & 0  \\
\end{array}\right).
\end{equation}
\end{subequations}
The Wronskian vector of this case is also in the form of
\begin{subequations}
\label{phi-breather}
\begin{equation}
\varphi=(\varphi_{1, 1}, \varphi_{1, 2}, \varphi_{2, 1}, \varphi_{2,
2}, \cdots, \varphi_{N, 1}, \varphi_{N, 2})^{T},
\end{equation}
and for convenience we set
\begin{equation}
\varphi^{+}=(\varphi_{1, 1}, \varphi_{2, 1}, \cdots, \varphi_{N,
1})^{T}, ~~ \varphi^{-}=(\varphi_{1, 2}, \varphi_{2, 2}, \cdots,
\varphi_{N, 2})^{T}. \label{entry-lim-br-ab }
\end{equation}
\end{subequations}
Substituting \eqref{phi-breather} into \eqref{mKdV-condition} with the above $\mathbb{A}$ and $\mathbb{B}$ we get
\begin{subequations}
\label{entry-condition}
\begin{equation}
\label{entry-condition-A}
 \varphi _{xx}^ +  = \mathbb{A}_N{\varphi ^ + }, ~
\varphi _{xx}^ - =\bar{\mathbb{A}}_N  {\varphi ^ - },
 \end{equation}
\begin{equation}
\label{entry-condition-B} \varphi _x^ +  =
\mathbb{B}_{N}\bar{\varphi} ^ - , ~
 \varphi _x^ -  = \bar{\mathbb{B}}_N \bar{\varphi }^ + ,
\end{equation}
\begin{equation}
\label{entry-condition-t}
\varphi^{\pm}_{t}=-4\varphi^{\pm}_{xxx},
\end{equation}
\end{subequations}
 where
\begin{subequations}
\label{mathbb-A-B}
\begin{equation}
\mathbb{A}_N=\left(\begin{array}{ccccccc}
k^2_1 & 0    & 0   & \cdots & 0   & 0 & 0 \\
2k_1   & k^2_1  & 0   & \cdots & 0   & 0 & 0\\
1   & 2k_1    & k^2_1 & \cdots & 0   & 0 & 0\\
\cdots &\cdots &\cdots &\cdots &\cdots &\cdots &\cdots \\
0   & 0    & 0   & \cdots & 1 & 2k_1  & k^2_1
\end{array}\right)_N,
\label{mathbb-A-31}
\end{equation}
\begin{equation}
\mathbb{B}_N=\left(\begin{array}{cccccc}
k_1 & 0    & 0   & \cdots & 0   & 0 \\
1   & k_1  & 0   & \cdots & 0   & 0 \\
\cdots &\cdots &\cdots &\cdots &\cdots &\cdots \\
0   & 0    & 0   & \cdots & 1   & k_1
\end{array}\right)_N.
\label{mathbb-B-31}
\end{equation}
\end{subequations}
From \eqref{entry-condition-A}, \eqref{entry-condition-t} and
\eqref{mathbb-A-31}, one first has
\begin{equation}
\label{gen-sol-mkdv-31} \varphi^+ =  \mathcal {A}\mathcal
{Q}^{+}_{0}+\mathcal{B}\mathcal {Q}^{-}_{0},~ \varphi^-= \mathcal
{C}\bar{\mathcal {Q}}^{+}_{0}+\mathcal{D}\bar{\mathcal
{Q}}^{-}_{0},~~ \mathcal{A}, \mathcal{B},\mathcal{C},\mathcal{D} \in
\widetilde{G}_N(\mathbb{C}),
\end{equation}
where $\mathcal{Q}^{\pm}_{0}$ are defined by \eqref{mathbb-Q}
together with $\xi_1$ defined in \eqref{xi-breather}.
Then substituting \eqref{gen-sol-mkdv-31} into
\eqref{entry-condition-B} and making using of
\eqref{gen-sol-relate-1}, one finds
\begin{equation}
\mathcal{C} =\bar{\mathcal{A}},~~\mathcal{D} =-\bar{\mathcal{B}}.
\label{gen-sol-LIMIT-breather}
\end{equation}
Thus,  \eqref{gen-sol-mkdv-31} reads
\begin{equation}
\label{gen-sol-mkdv-31-new} \varphi^+ =  \mathcal {A}\mathcal
{Q}^{+}_{0}+\mathcal{B}\mathcal {Q}^{-}_{0},~ \varphi^-=
\bar{\mathcal{A}}\bar{\mathcal
{Q}}^{+}_{0}-\bar{\mathcal{B}}\bar{\mathcal
{Q}}^{-}_{0},~~ \mathcal{A}, \mathcal{B} \in
\widetilde{G}_N(\mathbb{C}).
\end{equation}

The Wronskian $f(\varphi)$ with $\varphi$ defined by
\eqref{phi-breather} with \eqref{gen-sol-mkdv-31-new} will provide a limit solution of breathers for
the mKdV equation. In fact, quite similar to the procedure for the limit solutions of solitons we
described in Sec.\ref{sec:3.1},
here the Wronskian $f(\varphi)$ with  \eqref{phi-breather} is
related to the $N$-breather solution by taking the limit
$k_j\rightarrow k_1$ successively for $j=2,3,\cdots,N$.

\section{Rational solutions}

\subsection{Backgrounds}

Following solution structures of the KdV equation\cite{ZDJ-arxiv},
rational solutions should be led from zero eigenvalues of the coefficient matrix $\mathbb{A}$.
For the mKdV equation this requires $|\mathbb{B}|=0$.
However, this is not allowed in the Wronskian verification (see Appendix \ref{A-1})
because a trivial $\mathbb{B}$ will lead to a zero Wronskian in the light of \eqref{eqA1}.
Thus, it is clear that to get non-trivial rational solutions to the mKdV equation,
we need a non-trivial matrix $\mathbb{B}$.

Let us go back to the KdV-mKdV equation \eqref{kdv-mkdv}, i.e.,
\begin{equation}
{V_t}+ 12{v_0}V{V_X} + 6{V^2}{V_X} + {V_{XXX}} = 0,
\label{kdv-mkdv-1}
\end{equation}
which is related to the mKdV equation \eqref{mKdV} through a Galilean transformation
\begin{equation}
v(x,t)=v_{0}+V(X,t), ~~x=X+6{v_0}^{2}t, \label{GT}
\end{equation}
where $v_0$ is a real parameter.
We note that the equation \eqref{kdv-mkdv-1} admits non-trivial and non-singular rational solutions when $v_0\neq 0$.
Then using the transformation \eqref{GT}  rational solutions to the mKdV equation can be obtained.\
This fact has been realized via  B\"{a}cklund transformation(BT)\cite{Ono-JPSJ-1976}
and Hirota method with a limiting procedure\cite{Abl-Satsuma-1978,Ablowitz-book},
but the presentation for high order rational solutions is complicated.
In the following we will derive rational solutions in terms of  Wronskian,
which provides not only explicit but also impact forms for high order rational solutions.

\subsection{Rational solutions}

Still employing the same transformation as \eqref{trans-mKdV}, i.e.,
\begin{equation}
V=i\,\Big( \ln\frac{\bar{f}}{f}\Big)_X,
\label{trans-kdv-mkdv}
\end{equation}
 the KdV-mKdV equation \eqref{kdv-mkdv-1}
can be written into the bilinear form\cite{Wadati-JPSJ-1975,HS-PTP-1976}
\begin{subequations}
\begin{align}
(D_t+D_{X}^{3})\bar{f}\cdot f & =0,\label{eq4a}
\\
(D_{X}^{2}-2iv_{0}D_X)\bar{f}\cdot f & =0.\label{eq4b}
\end{align}
\label{bilinear-kdv-mkdv}
\end{subequations}

For the solutions to \eqref{bilinear-kdv-mkdv} in Wronskian form, we have
\begin{theorem}\label{Th 4.1}
The bilinear equation \eqref{bilinear-kdv-mkdv} admits Wronskian solution
\begin{equation}
f=W(\phi)=|\h{N-1}|,
\label{f-N}
\end{equation}
where the entry vector $\phi$ satisfies
\begin{subequations}
\label{cond}
\begin{align}
i\phi _{X}& =v_{0}\phi +B(t)\bar{\phi}, \label{eq99}
\\
\phi _{t}& =-4\phi _{XXX}+C(t)\phi, \label{eq98}
\end{align}
\end{subequations}
in which  $B(t)$ and $C(t)$ are two $N\times N$ matrices of $t$ but independent of $x$, and
satisfy
\begin{subequations}
\begin{eqnarray}
&&|B(t)|\neq 0,~~ \\
&&\mathrm{tr}{(C(t))}\in \mathbb{R}(t),\label{tr-Ct}\\
&&{B_t}(t) + B(t)\bar{C}(t)-C(t)B(t)=0.
\end{eqnarray}
\end{subequations}
\end{theorem}

The proof is similar to the one for Theorem \ref{Th 2.1}, but \eqref{eq99} results in  complicated expression for $\bar{f}$ and its derivatives.
We leave the proof in \ref{A:sec-4}.

As in Sec.\ref{sec:2.3}, with the help of the auxiliary matrix $\mathbb{A}$ and auxiliary  equation
\begin{equation}\label{A-aux}
\varphi_{XX}=\mathbb{A}\varphi,
\end{equation}
one can simplify the CES \eqref{cond} to
\begin{subequations}
\label{cond1}
\begin{align}
i\varphi _{X}& =v_{0}\varphi +\mathbb{B}\bar{\varphi}, \label{eq97}
\\
\varphi _{t}& =-4\varphi_{XXX},\label{eq96}
\end{align}
\end{subequations}
where both $\mathbb{A}$ and $\mathbb{B}$ are $N\times N$ complex constant matrices and  related by
\begin{equation}\label{A-BB}
    \mathbb{A}=\mathbb{B}\bar{\mathbb{B}}-v_0^{2}I_N,
\end{equation}
in which $I_N$ is the $N$th-order unit matrix.
It might be possible that here we discuss all possible solutions according to the eigenvalues of $\mathbb{A}$,
as we have done  in the previous section.
However, since we have had a clear description for solitons and breathers of the mKdV equation in the previous section,
and the parameter $v_0$ will bring more complexity,
in the following we can neglect the discussion of $\mathbb{A}$ and let us only focus on rational solutions.
In fact, rational solutions correspond to the zero eigenvalues of $\mathbb{A}$.

We derive rational solutions as the limit solutions of solitons.
The $N$-soliton solution corresponds to
\begin{equation}
\mathbb{B}=\mathrm{diag}\Big(-\sqrt{v_{0}^{2}+k_{1}^{2}},-\sqrt{v_{0}^{2}+k_{2}^{2}},\cdots,
-\sqrt{v_{0}^{2}+k_{N}^{2}}\ \Big),
\end{equation}
where we take $k_j$ to be $N$ distinct real positive numbers.
In this case, a solution to the CES \eqref{cond1} is
\begin{equation}
\varphi=(\varphi_{1},  \varphi_{2}, \cdots,  \varphi_{N})^{T},
\label{gen-s-1}
\end{equation}
with
\begin{equation}
\varphi_{j}=\sqrt{2v_{0}+2ik_{j}}\,e^{\eta_j}+\sqrt{2v_{0}-2ik_{j}}\,
e^{-\eta_j}, ~\eta_j=k_{j}X-4k_{j}^{3}t, \label{phi-j}
\end{equation}
This provides an $N$-soliton solution to the KdV-mKdV equation \eqref{kdv-mkdv-1}
through the transformation \eqref{trans-kdv-mkdv} with $f(\phi)$.

Now in the CES \eqref{eq97} we take $\mathbb{B}$ to be a lower triangular Toeplitz matrix,
\begin{equation}
\mathbb{B}=\left(
\begin{array}{cccc}
\alpha _{0} &  &  &  \\
\alpha _{1} & \alpha _{0} &  &  \\
\vdots  & \ddots  & \ddots  &  \\
\alpha _{N-1} & \cdots  & \alpha _{1} & \alpha _{0}
\end{array}
\right),
\end{equation}
with
\begin{equation}
\alpha_{j}=-\frac{1}{(2j)!}\frac{\partial ^{2j} }{{\partial
k_{1}}^{2j}}\sqrt{v_{0}^{2}+k_{1}^{2}}\,\Big|_{k_{1}=0},~~(j=0,1,\cdots,N-1).
\end{equation}
In this case the CES \eqref{cond1} admits a solution
\begin{subequations}
\begin{equation}
\psi =(\psi_1,\psi_2,\cdots,\psi_N)^T,
\end{equation}
with
\begin{equation}
\psi_{j+1}= \frac{1}{(2j)!}\frac{\partial ^{2j} }{{\partial
k_1}^{2j}}\varphi_{1}\,\biggr|_{k_1=0},~~(j=0,1,\cdots,N-1),
\end{equation}
\label{rat-psi}
\end{subequations}
where $\varphi_1$ is defined by \eqref{phi-j}. Then, with such a
$\psi$ as basic column vector, the Wronskian
\begin{equation}
f=W(\psi)=|\h{N-1}| \label{f-N}
\end{equation}
provides non-singular rational solutions to the KdV-mKdV equation \eqref{kdv-mkdv-1}.
A simplified form of these solutions is
\begin{equation}
V(X,t)=\frac{-2(F_{1,X}F_2-F_1F_{2,X})}{F_1^2+F_2^2},~~~F_1=\mathrm{Re}[f],~ F_2=\mathrm{Im}[f].
\end{equation}
We list the first three non-trivial $f$ (for $N=2,3,4$, respectively) in the following,
\begin{subequations}
\begin{align}
f =&4(2v_{0}X+i),\label{eq1012}
\\
f =&\frac{16v_{0}\sqrt{2v_{0}}}{3}
\Big[X^{3}+12t-\frac{3X}{4v_{0}^{2}}+\frac{3i}{2v_{0}}(X^{2}+\frac{1}{4v_{0}^{2}})\Big],\label{eq1013}
\\
f =&
-\frac{1}{{v_0^4}} - \frac{{4{X^2}}}{{v_0^2}} - \frac{{16}}{3}(12tX +{X^4}) - \frac{{64}}{{45}}v_0^2(720{t^2} - 60t{X^3}  - {X^6})\nonumber\\
&+ 4i{v_0}\Big( - \frac{{24t}}{{v_0^2}} +
\frac{X}{{v_0^4}} + 32t{X^2} +\frac{{16}}{{15}}{X^5}\Big).\label{f4}
\end{align}
\end{subequations}

For the rational solutions to the mKdV equation \eqref{mKdV}, we have
\begin{proposition}
\label{P:rat-mkdv}
The non-trivial rational solutions to the mKdV equation is given by
\begin{equation}
v(x,t)=v_{0}-\frac{2(F_{1,X}F_2-F_1F_{2,X})}{F_1^2+F_2^2},~~ X=x-6v_{0}^{2}t, \label{rat}
\end{equation}
where $f$ is the Wronskian \eqref{f-N} composed with \eqref{rat-psi}, $F_1=\mathrm{Re}[f],~ F_2=\mathrm{Im}[f]$.
\end{proposition}

The first non-trivial rational solution to the mKdV equation reads
\begin{equation}
v =v_{0}-\frac{4v_{0}}{4v_{0}^{2}X^{2}+1},~~X=x-6v_{0}^{2}t,
\label{1rs}
\end{equation}
and the second one
\begin{equation}
v=v_{0}-\frac{12v_{0}(X^{4}+\frac{3}{2v_{0}^{2}}X^{2}-\frac{3}{16v_{0}^{4}}
-24Xt)}{4v_{0}^{2}(X^{3}+12t-\frac{3X}{4v_{0}^{2}})^{2}+9(X^{2}+
\frac{1}{4v_{0}^{2}})^{2}},~~~X=x-6v_{0}^{2}t. \label{2rs}
\end{equation}

We note that there is a bilinear BT\cite{Nimmo-Freeman-JPA} related to \eqref{bilinear-kdv-mkdv}
and the BT admits multi-soliton solutions\cite{Nimmo-Freeman-JPA} and rational solutions\cite{ZH-S-rational}
in Wronskian form.
Besides, rational solutions of the mKdV equation (for both $\varepsilon=\pm 1$)
may also be derived (through the Galilean transformed version)
by using the long-wave-limit approach described in Ref.\cite{rational-1} (also see \cite{Gesztesy-Rmp-1991,Zhang-Toda rational}).

In next section we will discuss dynamics of solutions including these rational solutions.

\section{Dynamics analysis}

In this section we investigate dynamics of two-soliton solutions,
limit solutions, breathers and rational solutions. To describe the
relationship between two-soliton solution and the simplest limit solution,
let us start from the asymptotic behaviors of two-soliton
interactions.

\subsection{Dynamics of solitons}

Soliton solutions to the mKdV equation \eqref{mKdV} can be described by
\begin{equation}
v(x,t)=-\frac{2(F_{1,x}F_2-F_1F_{2,x})}{F_1^2+F_2^2},~~~F_1=\mathrm{Re}[f],~ F_2=\mathrm{Im}[f],
\label{solu-mkdv}
\end{equation}
where $f=f(\varphi)=|\widehat{N-1}|$ is the Wronskian composed by
the basic column vector $\varphi$ which is defined by
\eqref{gen-solitons}, or equivalently, by either \eqref{mkdv-c12} or
\eqref{mkdv-c13} in Case I.
We note that from the transformation \eqref{trans-mKdV} the solution to the mKdV equation \eqref{mKdV} can also be written as
\begin{equation}
v=-2\Big(\arctan \frac{F_1}{F_2}\Big)_x=2\Big(\arctan \frac{F_2}{F_1}\Big)_x,~~
F_1=\mathrm{Re}[f],~ F_2=\mathrm{Im}[f],
\label{sol-arctan}
\end{equation}
while \eqref{solu-mkdv} gives a more explicit form.

In the following we investigate 1- and
2-soliton solutions, which are, respectively, corresponding to
\begin{subequations}
\begin{align}
f=f_1=& f(\varphi_1)=\varphi_1,\label{f-1ss}\\
f=f_2=& f((\varphi_1,\varphi_2)^T)=
\Big|\begin{array}{cc}
\varphi_1&\varphi_{1,x}\\
\varphi_2&\varphi_{2,x}
\end{array}\Big|,\label{f-2ss}
\end{align}
\end{subequations}
with   $\varphi_j$ defined by \eqref{mkdv-c11}, i.e.,
\begin{equation}
\varphi^{}_{j}= a_{j}^+  e^{\xi_{j}}+ i a_{j}^-
 e^{-\xi_{j}}, ~\xi_{j}=k_{j}x-4k_{j}^{3}t+\xi_{j}^{(0)},~~
 a_{j}^\pm, k_j, \xi_{j}^{(0)} \in \mathbb{R}.
\label{mkdv-c11-sec5}
\end{equation}

Then, one-soliton solution to the  mKdV equation \eqref{mKdV} reads
\begin{equation}
\label{1-soliton} v=-2\cdot
\mathrm{sgn}\Big[\frac{a_1^-}{a_1^+}\Big]\cdot k_1\cdot
\mbox{sech}\Big(2k_1x-8k_1^3t-\ln\Big|\frac{a_1^-}{a_1^+}\Big|\Big),
\end{equation}
as depicted in Fig.\ref{Fig-1},
where for convenience, we call (a) soliton and (b)
anti-soliton due to the signs of their amplitudes.

\begin{figure}[!h]
\setlength{\unitlength}{1mm}
\begin{center}
\begin{picture}(00,40)
\put(-65,-5){\resizebox{!}{3.7cm}{\includegraphics{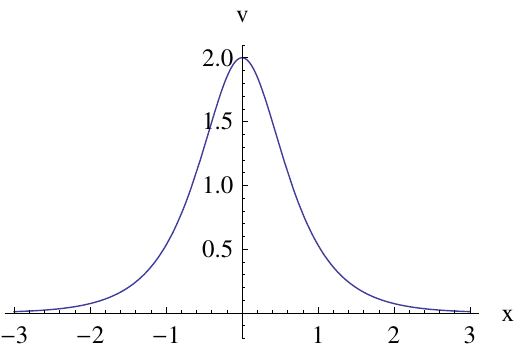}}}
\put(5,-5){\resizebox{!}{3.7cm}{\includegraphics{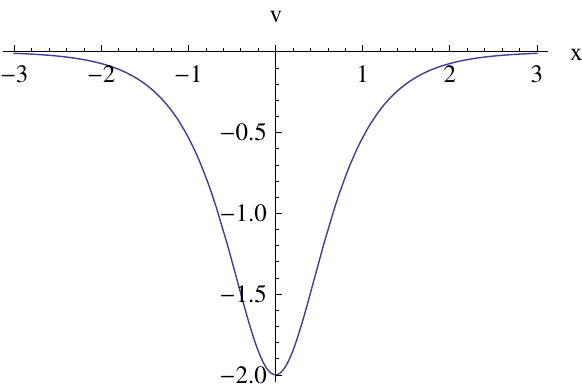}}}
\end{picture}
\end{center}
\vskip 1pt
\begin{minipage}{15cm}{\footnotesize
~~~~~~~~~~~~~~~~~~~~~~~~~~~~~~~~~~~~(a)~~~~~~~~~~~~~~~~~~~~~~~~~~~~~~~~~~~~
~~~~~~~~~~~~~~~~~~~~~~~~(b)}
\end{minipage}
\caption{ The shape of  one-soliton solution given by
\eqref{1-soliton} at $t=0$. (a). Soliton for
$a_1^+=a_1^-=1$, $k_1=-1$ and $\xi_{1}^{(0)}=0$. (b). Anti-soliton for $a_1^+=a_1^-=1$,
$k_1=1$  and $\xi_{1}^{(0)}=0$. \label{Fig-1}}
\end{figure}

Obviously, this soliton is identified by the amplitude
\[\mathrm{Amp}= -2\cdot \mathrm{sgn}\big[\frac{a_1^-}{a_1^+}\big]\cdot k_1\]
and top trace (trajectory)
\begin{equation}
x(t)=4k_1^2t+\frac{1}{2k_1}\ln \big|\frac{a_1^-}{a_1^+}\big|,
\label{top}
\end{equation}
or velocity $4k_1^2$.
Obviously, solitons of the mKdV equation are single-direction waves.

Next, let us look at two-soliton solution.
The two-soliton solution of the mKdV (\ref{mKdV}) can be expressed by \eqref{solu-mkdv}
where from \eqref{f-2ss}
\begin{subequations}
\begin{align}
F_1&=(k_2-k_1)\big(a_1^-a_2^-e^{4(k_1^3+k_2^3)t-(k_1+k_2)x}+a_1^+a_2^+e^{-4(k_1^3+k_2^3)t+(k_1+k_2)x}\big),\\
F_2&=(k_2+k_1)\big(a_1^-a_2^+e^{4(k_1^3-k_2^3)t-(k_1-k_2)x}-a_1^+a_2^-e^{4(-k_1^3+k_2^3)t+(k_1-k_2)x}\big).
\end{align}
\end{subequations}
We assume   that $a_j^\pm \neq 0$ and set
$\mathrm{sgn}\big[\frac{a_1^{-}}{a_1^{+}}\big]=\mathrm{sgn}\big[\frac{a^{-}_2}{a^+_2}\big]$,
so that  $F_1,F_2$ are not zero at same time.
Without loss of generality, we also take  $0<|k_2|<|k_1|$.
For the analysis of asymptotic behaviors, it is convenient to use
the following expression
\begin{equation}
v=2\Big(\arctan\,\frac{F_2}{F_1}\Big)_x.
\label{2-two-soliton}
\end{equation}

There are two types of 2-soliton interactions, soliton-soliton (or
anti-soliton-anti-soliton)
 interaction and soliton-anti-soliton interaction, as shown in Fig.\ref{Fig-2}.

\begin{figure}[!h]
\setlength{\unitlength}{1mm}
\begin{center}
\begin{picture}(00,40)
\put(-65,-5){\resizebox{!}{4cm}{\includegraphics{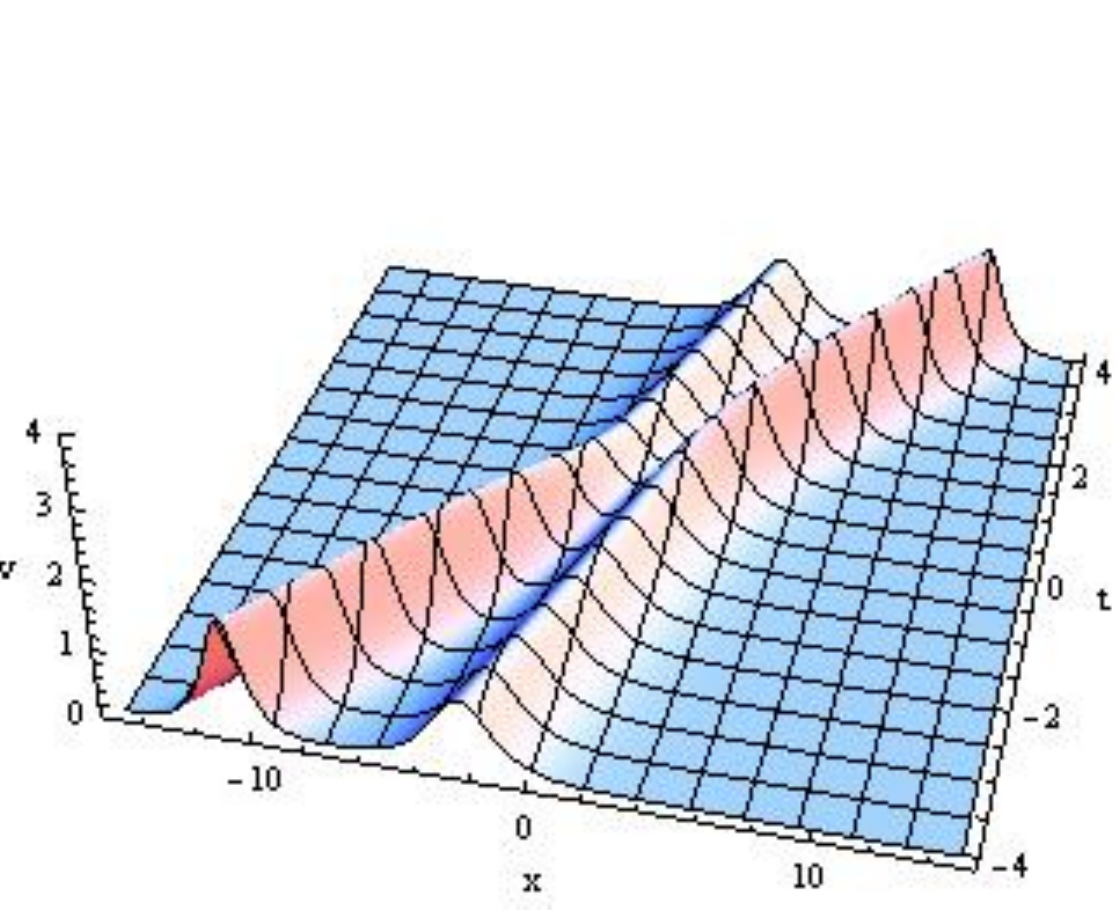}}}
\put(5,-5){\resizebox{!}{4cm}{\includegraphics{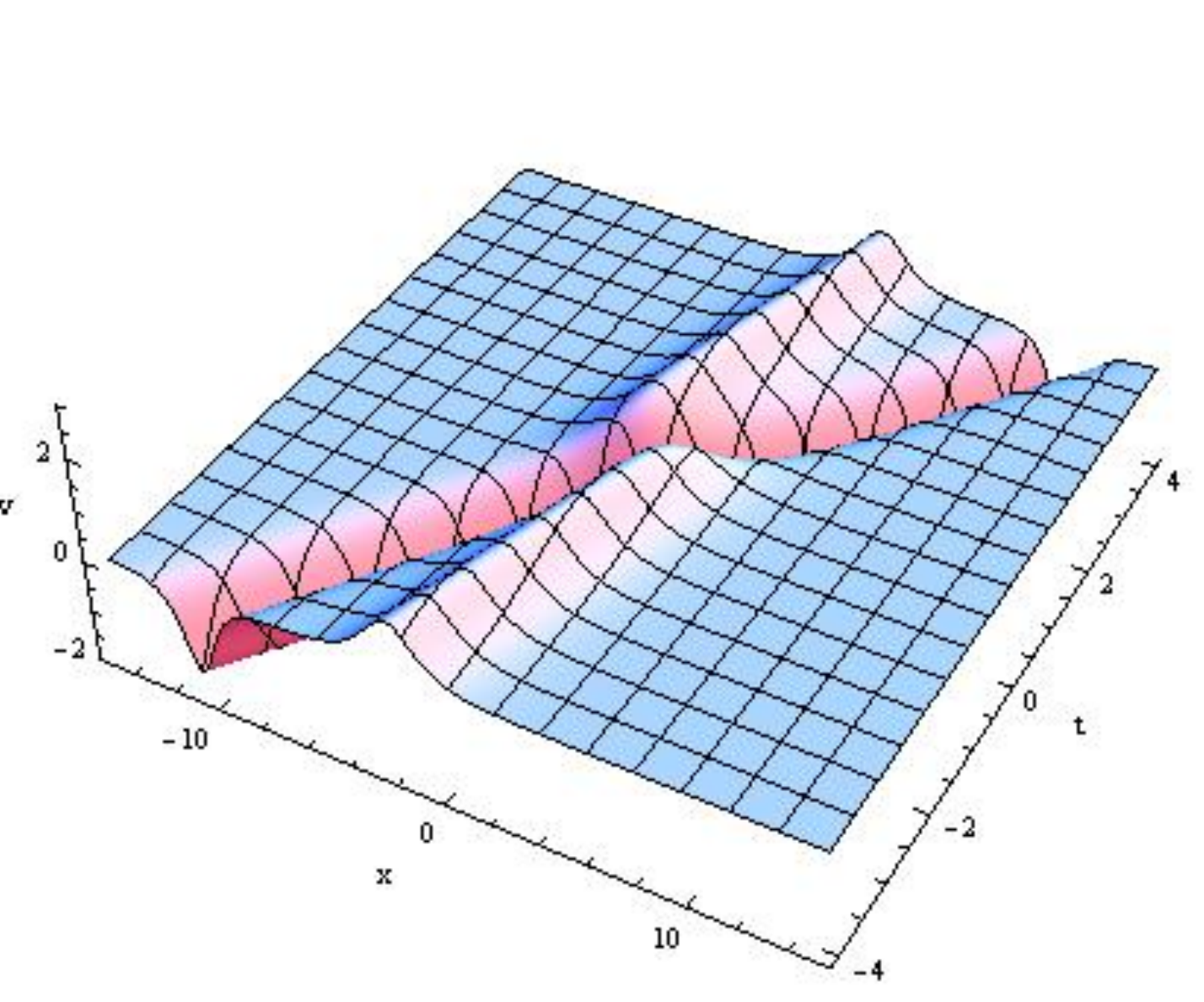}}}
\end{picture}
\end{center}
\vskip 1pt
\begin{minipage}{15cm}{\footnotesize
~~~~~~~~~~~~~~~~~~~~~~~~~~~~~~~~~~~~(a)~~~~~~~~~~~~~~~~~~~~~~~~~~~~~~~~~~~~
~~~~~~~~~~~~~~~~~~~~~~~~(b)}
\end{minipage}
\caption{Two-soliton interactions. (a). Soliton-soliton for
$a_1^+=a_1^-=a_2=1,~b_2=-1$, $k_1=0.8,~k_2=0.5$  and $\xi_{1}^{(0)}=\xi_{2}^{(0)}=0$. (b).
Soliton-anti-soliton for $a_1^+=a_1^-=a_2=1,~b_2=-1$,
$k_1=-0.8,~k_2=0.5$   and $\xi_{1}^{(0)}=\xi_{2}^{(0)}=0$. \label{Fig-2}}
\end{figure}

To investigate asymptotic behaviors of the two solitons involved in
interaction, we first name them $k_1$-soliton and $k_2$-soliton,
respectively. Then we rewrite the two-soliton solution
\eqref{2-two-soliton} in terms of the following coordinates,
\begin{equation}
\big(X_1=x-4k_1^2t, t \big), \label{k1-frame}
\end{equation}
which then gives
\begin{equation}
\label{2s-X}
v=2\Big(\arctan\,\frac{(a_1^+a_2^-e^{2k_1X_1}-a_2^+a_1^-e^{8k_2(k_1^2-k_2^2)t+2k_2X_1})(k_1+k_2)}
{(a_1^-a_2^-+a_1^+a_2^+e^{8k_2(k_1^2-k_2^2)t+2(k_1+k_2)X_1})(k_1-k_2)}
\Big)_{X_1}.
\end{equation}
Noting that for any $0<|k_2|<|k_1|$ it is always valid that
$k_1^2-k_2^2>0$ and $\frac{k_1+k_2}{k_1-k_2}>0$, we can keep $X_1$
to be constant and let $t$ go to infinity. Then we can find there is
only $k_1$-soliton left along the line $X_1=const.$ and also find
how the $k_1$-soliton is asymptotically identified by its top trace
and amplitude, for both $t\to \pm\infty$.

As for details, when $k_2>0,~t \to -\infty$ or $k_2<0,~t \to +\infty$, i.e.
$\mathrm{sgn}[k_2]\cdot t \to -\infty$, the solution \eqref{2s-X} becomes
\begin{align}
\label{2s-X-1}
v=&2\Big(\arctan
\frac{a_1^+(k_1+k_2)e^{2k_1X_1}}{a_1^-(k_1-k_2)}\Big)_{X_1}\nonumber\\
=&2\cdot\mathrm{sgn}\Big[\frac{a_1^-}{a_1^+}\Big]\cdot
k_1 \cdot
\mbox{sech}\Big(2k_1X_1-\ln\Big|\frac{a_1^-}{a_1^+}\Big|+\ln\frac{k_1+k_2}{k_1-k_2}\Big);
\end{align}
and when $k_2>0,~t \to +\infty$ or $k_2<0,~t \to -\infty$, i.e.
$\mathrm{sgn}[k_2]\cdot t \to +\infty$,  \eqref{2s-X} becomes
\begin{align}
\label{2s-X-2}
v=&2\Big(\arctan \frac{-a_1^-(k_1+k_2)}{a_1^+e^{2k_1X}(k_1-k_2)}
\Big)_{X_1} \nonumber\\
=&2\cdot\mathrm{sgn}\Big[\frac{a_1^-}{a_1^+}\Big]\cdot k_1 \cdot
\mbox{sech}\Big(2k_1X_1-\ln\Big|\frac{a_1^-}{a_1^+}\Big|-\ln\frac{k_1+k_2}{k_1-k_2}\Big).
\end{align}

We can also rewrite the two-soliton solution \eqref{2-two-soliton} in
terms of the coordinates
\begin{equation}
\big(X_2=x-4k_2^2t, t \big), \label{k2-frame}
\end{equation}
and do a similar asymptotic analysis for the $k_2$-soliton.
Finally, we reach to
\begin{theorem}\label{T:2s}
Suppose that
$\mathrm{sgn}\big[\frac{a_1^-}{a_1^+}\big]=\mathrm{sgn}\big[\frac{a_2^-}{a_2^+}\big]$,
$a_j^\pm \neq 0$ and $0<|k_2|<|k_1|$ in \eqref{2s-X}. Then,
when $\mathrm{sgn}[k_2]\cdot t \to \pm \infty$, the $k_1$-soliton
asymptotically follows
\begin{subequations}
\begin{align}
\mathrm{top~ trace:}~~& x(t)=4k_1^2t+\frac{1}{2k_1}\ln\Big|\frac{a_1^-}{a_1^+}\Big|\pm \frac{1}{2k_1}\ln\frac{k_1+k_2}{k_1-k_2},\\
\mathrm{amplitude:}~~& 2 \cdot
\mathrm{sgn}\Big[\frac{a_1^-}{a_1^+}\Big]\cdot k_1,
\end{align}
\end{subequations}
and when $\mathrm{sgn}[k_1]\cdot t \to \pm \infty$, the $k_2$-soliton asymptotically
follows
\begin{subequations}
\begin{align}
\mathrm{top~ trace:}~~& x(t)=4k_2^2t+\frac{1}{2k_2}\ln\Big|\frac{a_2^-}{a_2^+}\Big|\pm \frac{1}{2k_2}\ln\frac{k_1+k_2}{k_1-k_2},\\
\mathrm{amplitude:}~~& -2 \cdot
\mathrm{sgn}\Big[\frac{a_2^-}{a_2^+}\Big]\cdot k_2.
\end{align}
\end{subequations}
The phase shift for the $k_j$-soliton after interactions is
$\frac{1}{k_j}\ln\frac{k_2+k_1}{k_1-k_2}$.
\end{theorem}

Now it is completely clear how the
two-soliton interactions are related to the parameters $\{k_j,a_j,b_j\}$.
This will  be helpful to understand the asymptotic behavior of limit solutions.

\subsection{Asymptotic behavior of limit solutions}

The simplest limit solution in Case II is
\begin{equation}
\label{2-limit} v=2\Big(\arctan
\frac{a_1^+a_1^-(-48k_1^3t+4k_1x)}{{a_1^-}^2e^{8k_1^3t-2k_1x}+{a_1^+}^2e^{-8k_1^3t+2k_1x}}
\Big)_x.
\end{equation}
This is derived from \eqref{sol-arctan} with
\begin{subequations}
\begin{equation}
f=
\Big|\begin{array}{ll}
\varphi_1&~~\varphi_{1,x}\\
\partial_{k_1}\varphi_1&~~(\partial_{k_1}\varphi_{1})_{x}
\end{array}\Big|,\label{f-lim}
\end{equation}
where $\varphi_1$ is defined in \eqref{mkdv-c11}, i.e.,
\begin{equation}
\varphi^{}_{1}= a_{1}^+  e^{\xi_{1}}+ i a_{1}^-
 e^{-\xi_{1}}, ~\xi_{1}=k_{1}x-4k_{1}^{3}t+\xi_{1}^{(0)},~~
 a_{1}^{\pm}, k_1, \xi_{1}^{(0)} \in \mathbb{R}.
\end{equation}
\end{subequations}

\begin{figure}[!h]
\setlength{\unitlength}{1mm}
\begin{center}
\begin{picture}(00,40)
\put(-65,-5){\resizebox{!}{4.0cm}{\includegraphics{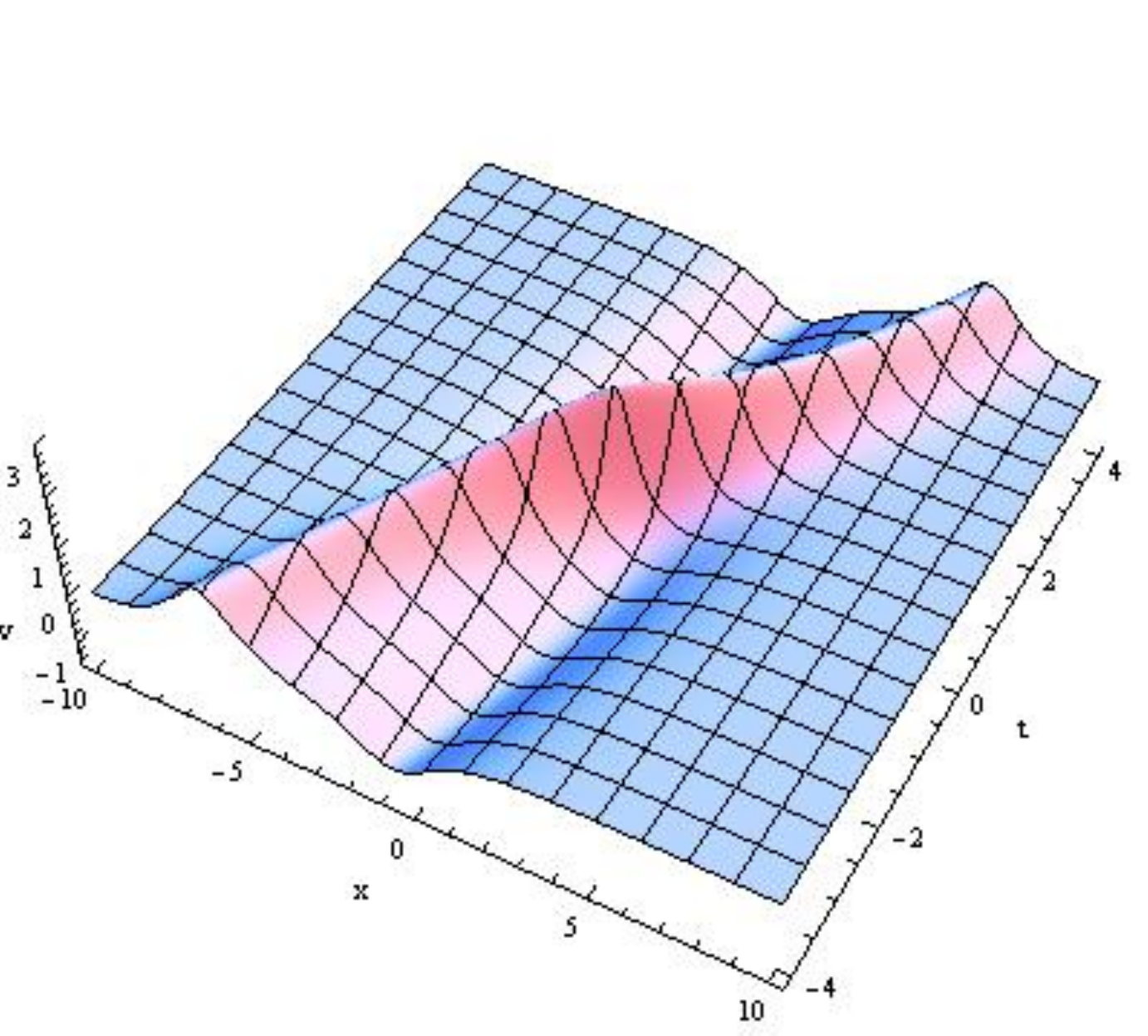}}}
\put(5,-5){\resizebox{!}{3.7cm}{\includegraphics{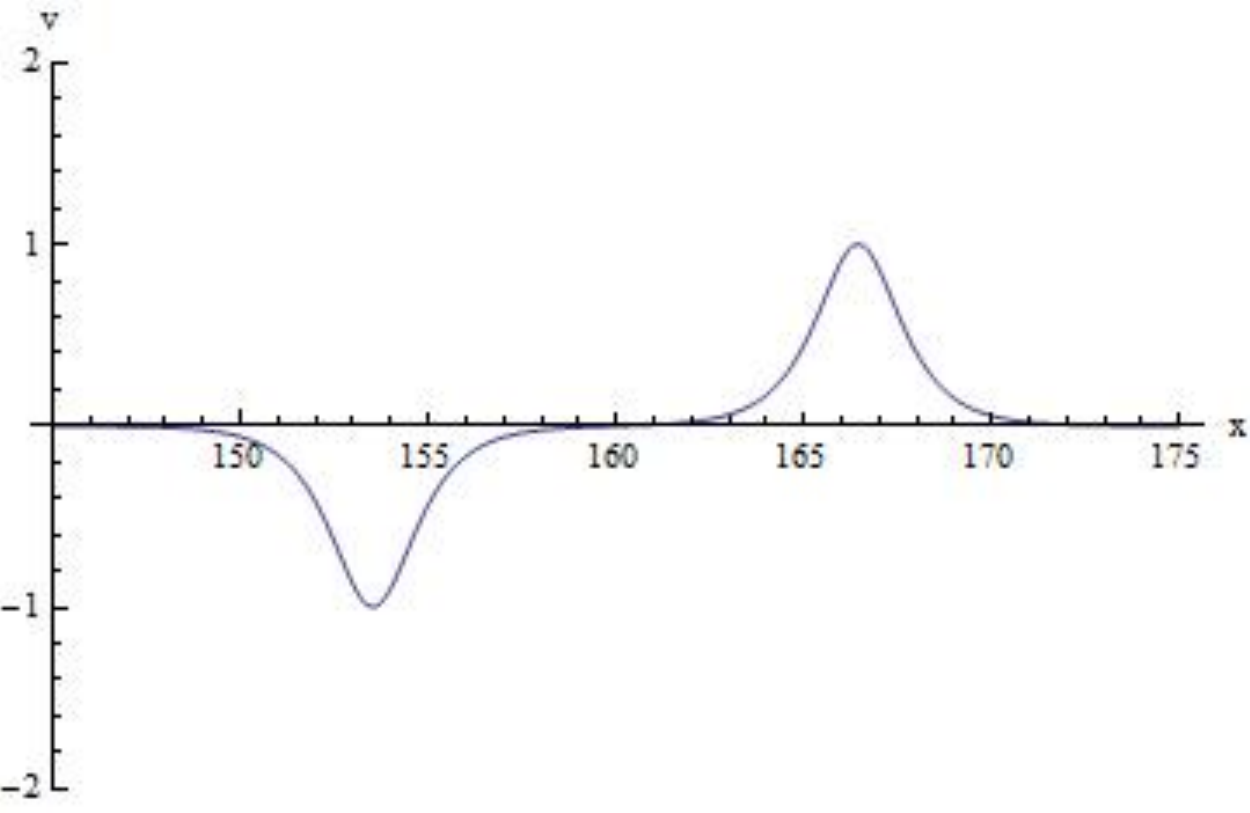}}}
\end{picture}
\end{center}
\vskip 1pt
\begin{minipage}{15cm}{\footnotesize
~~~~~~~~~~~~~~~~~~~~~~~~~~~~~~~~~~~~(a)~~~~~~~~~~~~~~~~~~~~~~~~~~~~~~~~~~~~
~~~~~~~~~~~~~~~~~~~~~~~~(b)}
\end{minipage}
\caption{ Limit solution given by \eqref{2-limit} for
$a_1^+=a_1^-=1$, $k_1=0.5$ and $\xi_{1}^{(0)}=0$. (a). Shape and motion. (b).
Asymmetric wave shape at $t=160$. \label{Fig-3}}
\end{figure}

The solution \eqref{2-limit} is depicted in Fig.\ref{Fig-3}.
We characterize dynamics of the solution by the following two points, which are typically different from the
interaction of two normal solitons that we described in the previous subsection.
These two points are
\begin{itemize}
\item{Soliton-anti-soliton interaction with (asymptotically) asymmetric wave shape,}
\item{Top trace of each soliton is asymptotically governed by logarithm and linear functions.}
\end{itemize}

The first point can be explained as follows. Recall the two-soliton
interaction with $a_2^{\pm}=a_1^{\pm}\neq 0$ and $k_1\cdot k_2>0$.
According to Theorem \ref{T:2s}, this is a soliton-anti-soliton
interaction and the absolute of amplitude of each soliton is
$2|k_j|$. Obviously, the asymmetric shape of limit solution
\eqref{2-limit} coincides well with the limit $k_2 \to k_1$.

To understand the second point, again we
put the limit solution in the coordinate system
\begin{equation}
\big(Y=x-4k_1^2t, t \big), \label{k1-frame-lim}
\end{equation}
and this gives
\begin{equation}
\label{2-lim-Y} v=2\Big(\arctan\,\frac{4a_1^+a_1^-e^{2k_1
Y}k_1(Y-8k_1^2t)}{{a_1^-}^2+{a_1^+}^2e^{4k_1 Y}} \Big)_Y,
\end{equation}
which is described in Fig.\ref{Fig-4}.

\begin{figure}[!h]
\setlength{\unitlength}{1mm}
\begin{center}
\begin{picture}(00,40)
\put(-65,-5){\resizebox{!}{4.0cm}{\includegraphics{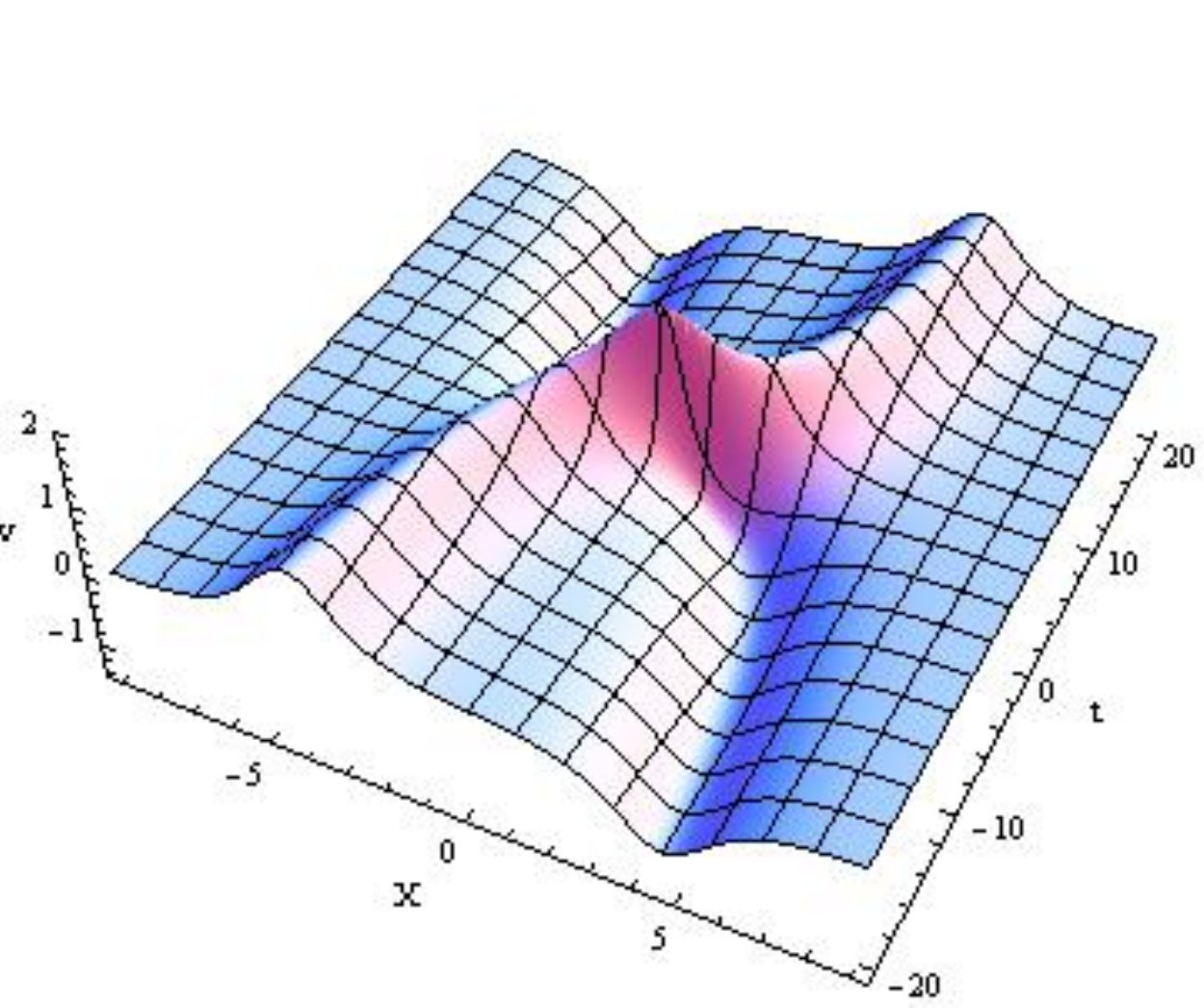}}}
\put(5,-5){\resizebox{!}{3.7cm}{\includegraphics{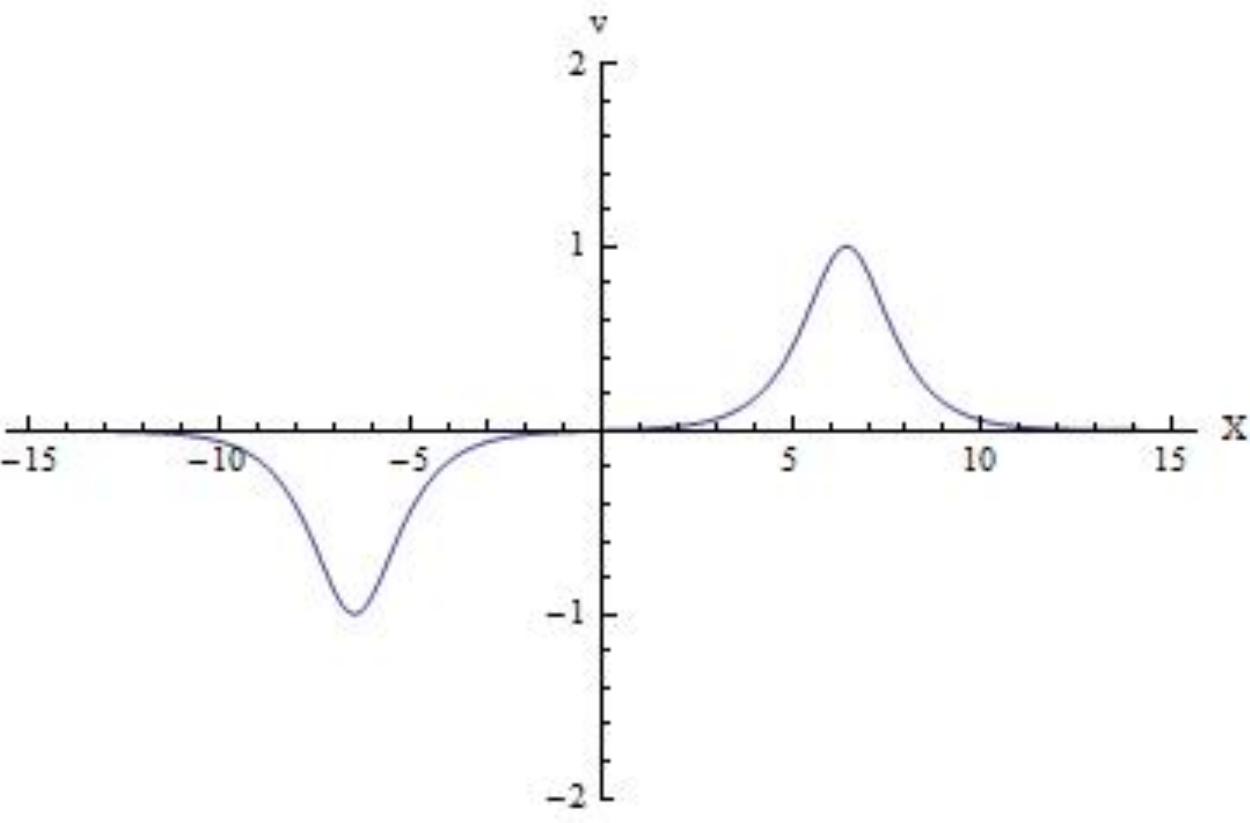}}}
\end{picture}
\end{center}
\vskip 1pt
\begin{minipage}{15cm}{\footnotesize
~~~~~~~~~~~~~~~~~~~~~~~~~~~~~~~~~~~~(a)~~~~~~~~~~~~~~~~~~~~~~~~~~~~~~~~~~~~
~~~~~~~~~~~~~~~~~~~~~~~~(b)}
\end{minipage}
\caption{Limit solution given by \eqref{2-lim-Y} for $a_1^+=a_1^-=1$
and $k_1=0.5$. (a). Shape and motion. (b). Asymmetric wave shape at
$t=160$. \label{Fig-4}}
\end{figure}

For convenience here we suppose that $k_1>0$.
Then by analyzing the leading terms as $t\to \pm\infty$ in the numerator and
denominator in \eqref{2-lim-Y}, we can  conclude the
asymptotic behaviors of the limit solution \eqref{2-lim-Y} as follows.

\begin{theorem}\label{T:2l}
Suppose that $a_1^{\pm}\neq 0$ and $k_1>0$ in \eqref{2-lim-Y}. When $t \to -\infty$ there are two waves moving with
amplitudes $\mp 2 \cdot \mathrm{sgn}\big[\frac{a^-_1}{a^+_1}\big]\cdot k_1$ and top traces
governed by the logarithm functions
\begin{subequations}
\begin{equation}
Y_{\pm}=\frac{1}{2k_1}\Big(\ln
\Big|\frac{a_1^-}{a_1^+}\Big|\pm \ln (-32k_1^3t)\Big),
\end{equation}
where the subscript $\pm$ of $Y$ stands for $Y\to \pm\infty$.
%
When $t \to +\infty$, there are also two waves moving with amplitudes
$\pm 2 \cdot \mathrm{sgn}\big[\frac{a^-_1}{a^+_1}\big]\cdot k_1$
and top traces
\begin{equation}
Y_{\pm}=\frac{1}{2k_1}\Big(\ln
\Big|\frac{a_1^-}{a_1^+}\Big|\pm \ln (32k_1^3t)\Big).
\end{equation}
\end{subequations}
\end{theorem}

Replacing $Y$ by $x$ by using \eqref{k1-frame-lim} the top traces of the waves in
Fig.\ref{Fig-3} are then asymptotically governed by linear and logarithm functions of $t$.

Let us end up this subsection with the following remark. With regard
to the limit solutions the top traces (or waves trajectories)  are
governed by (linear and) logarithm functions should be a typical
characteristic, which differs from normal soliton interactions with
straight line trajectories. (See
\cite{Zhang-Hietarinta,ZZZ-2009-PLA,ZDJ-TMP,KZS-sg} for more
examples).

\subsection{Breathers}
\label{sec:5.3}

Wronskian entries in Case III provides breather solutions to the mKdV
 equation \eqref{mKdV}. The simplest one corresponds to
\begin{subequations}
\begin{equation}
f=
\Big|\begin{array}{ll}
\varphi_{11}&~~\varphi_{11,x}\\
\varphi_{12}&~~\varphi_{12,x}
\end{array}\Big|=F_1+iF_2,
\label{f-breather}
\end{equation}
where
\begin{align}
\varphi_{11}&= a_{1}  e^{\xi_1}+b_{1} e^{-{\xi}_1}, ~
\varphi_{12}=\bar{a}_{1}  e^{\bar{\xi}_1} -\bar{b}_{1}
e^{-{\bar{\xi}_1}},\\
\xi_1&=k_{1}x- 4 k_1^3 t,~k_1, a_{1}, b_{1} \in \mathbb{C},
\label{xi-1-breather}
\end{align}
\begin{align}
F_1=&4k_{11}(a_{11}b_{11}+a_{12}b_{12})\cos(24k_{11}^2k_{12}t-8k_{12}^3t-2k_{12}x)\nonumber\\
~&+4k_{11}(a_{12}b_{11}-a_{11}b_{12})\sin(24k_{11}^2k_{12}t-8k_{12}^3t-2k_{12}x),
\label{F1-B}\\
F_2=&-2k_{12}e^{-2k_{11}(4k_{11}^2t+12k_{12}^2t+x)}\big[(b_{11}^2+b_{12}^2)e^{16k_{11}^3t} \nonumber\\
&+(a_{11}^2+a_{12}^2)e^{4k_{11}(12k_{12}^2t+x)}\big], \label{F2-B}
\end{align}
and we have written
\[k_1=k_{11}+ik_{12},~a_1=a_{11}+ia_{12},~b_1=b_{11}+ib_{12}.\]
\end{subequations}
Then the breather solution is expressed by \eqref{solu-mkdv} with
the above $F_1, F_2$.
We further assume that
\[\sin
\theta=\frac{a_{11}b_{11}+a_{12}b_{12}}{\alpha},~\alpha=\sqrt{(a_{11}b_{11}+a_{12}b_{12})^2+(a_{12}b_{11}-a_{11}b_{12})^2},
\]
and then rewrite \eqref{solu-mkdv} as\footnote{From \eqref{solu-breather} we can see that $k_{1j}\neq 0$ is necessary for getting nontrivial breather solutions.}
\begin{subequations}\label{solu-breather}
\begin{equation}
v=-2\Big(\arctan
\frac{P}
{Q}\Big)_x,
\label{solu-breather-5.22}
\end{equation}
where
\begin{align}
P&=2k_{11}\alpha\sin\big[2k_{12}(x-4t(3k_{11}^2-k_{12}^2))-\theta\big],\\
Q&=k_{12}
\big[(a_{11}^2+a_{12}^2)e^{2k_{11}(x+4t(3k_{12}^2-k_{11}^2))}
+(b_{11}^2+b_{12}^2)e^{-2k_{11}(x+4t(3k_{12}^2-k_{11}^2))}
\big].
\end{align}
\end{subequations}
Such a breather is described in Fig.\ref{Fig-5-1}.

\begin{figure}[!h]
\setlength{\unitlength}{1mm}
\begin{center}
\begin{picture}(00,40)
\put(-78,-5){\resizebox{!}{3.8cm}{\includegraphics{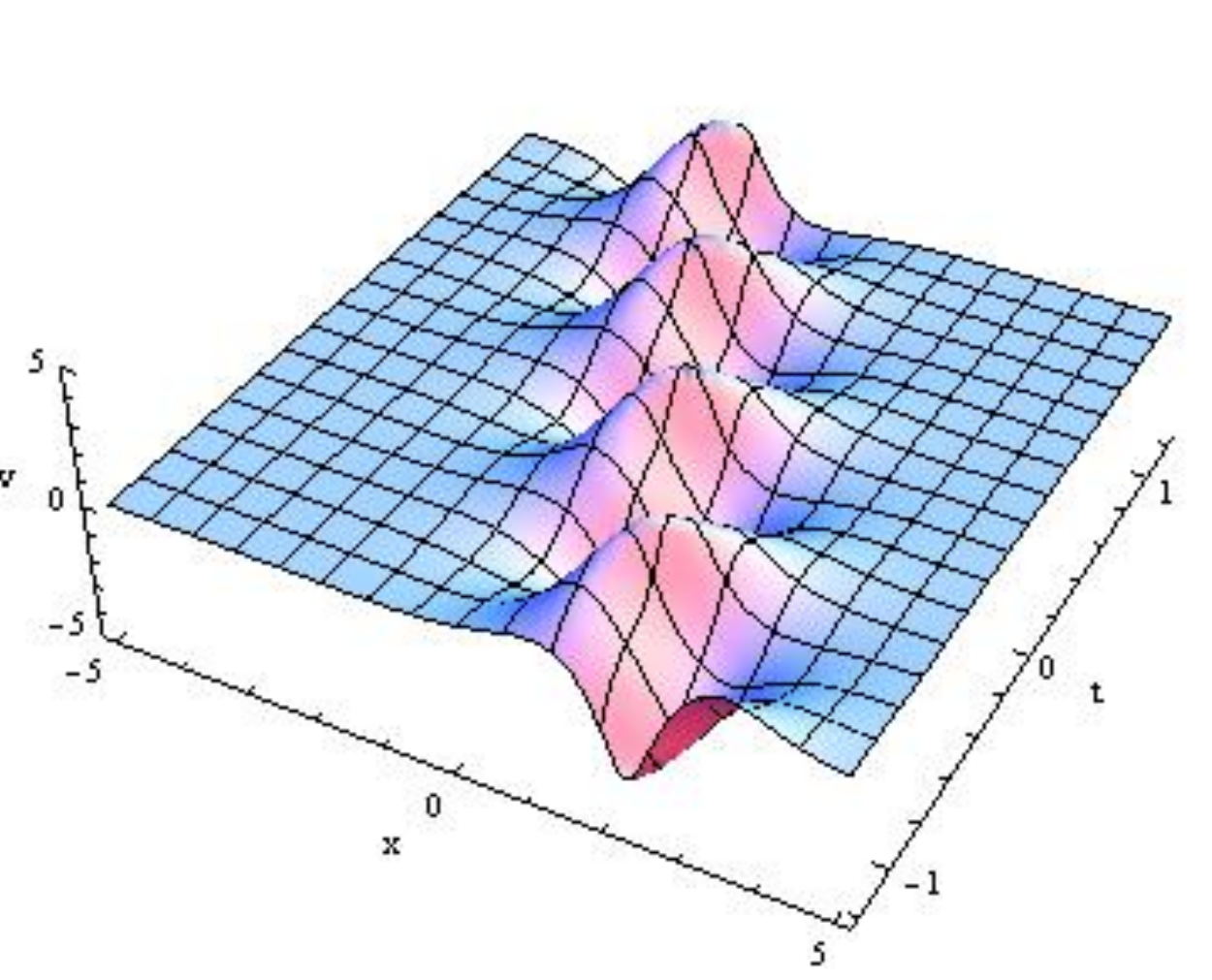}}}
\put(-27,-5){\resizebox{!}{3.8cm}{\includegraphics{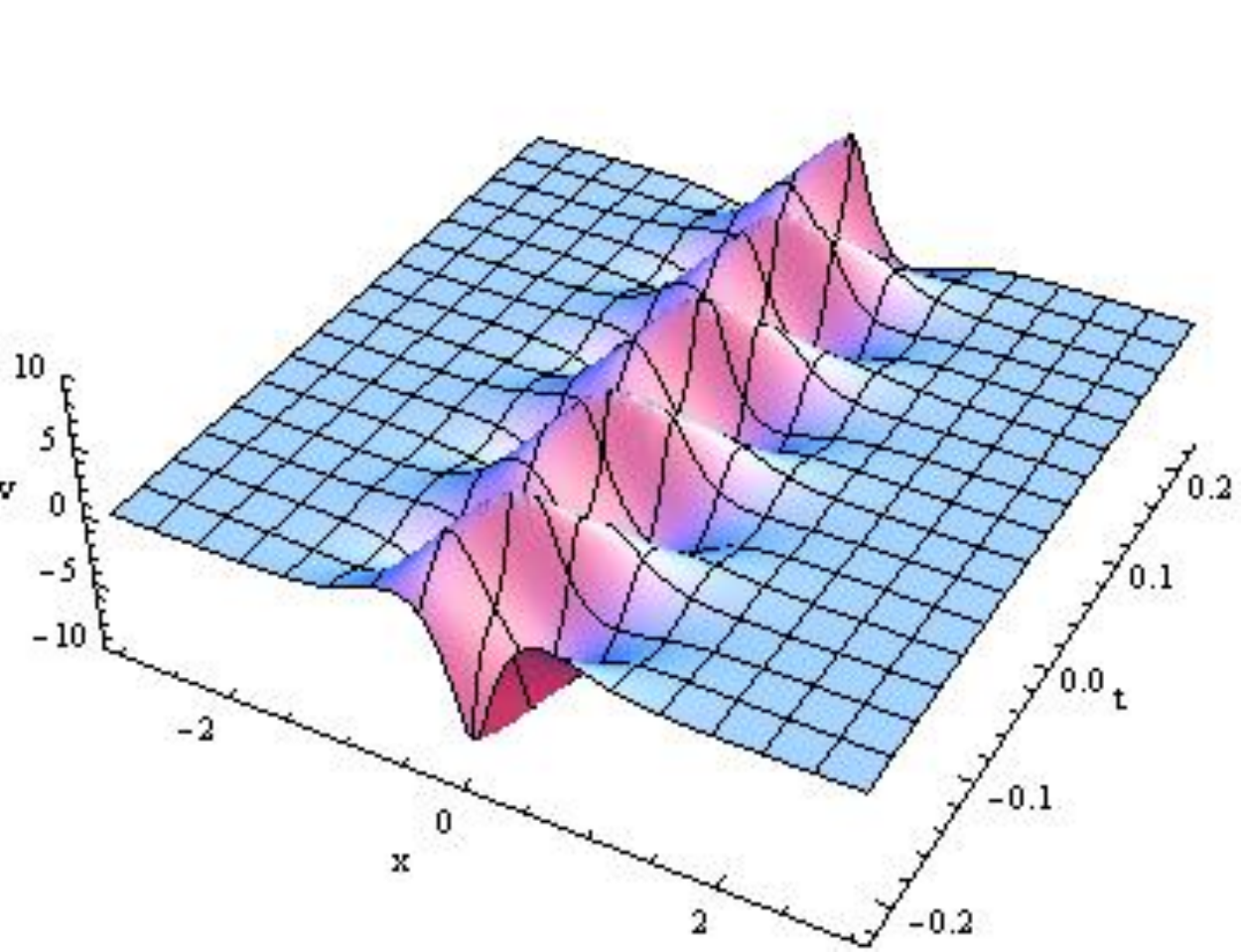}}}
\put(25,-5){\resizebox{!}{3.5cm}{\includegraphics{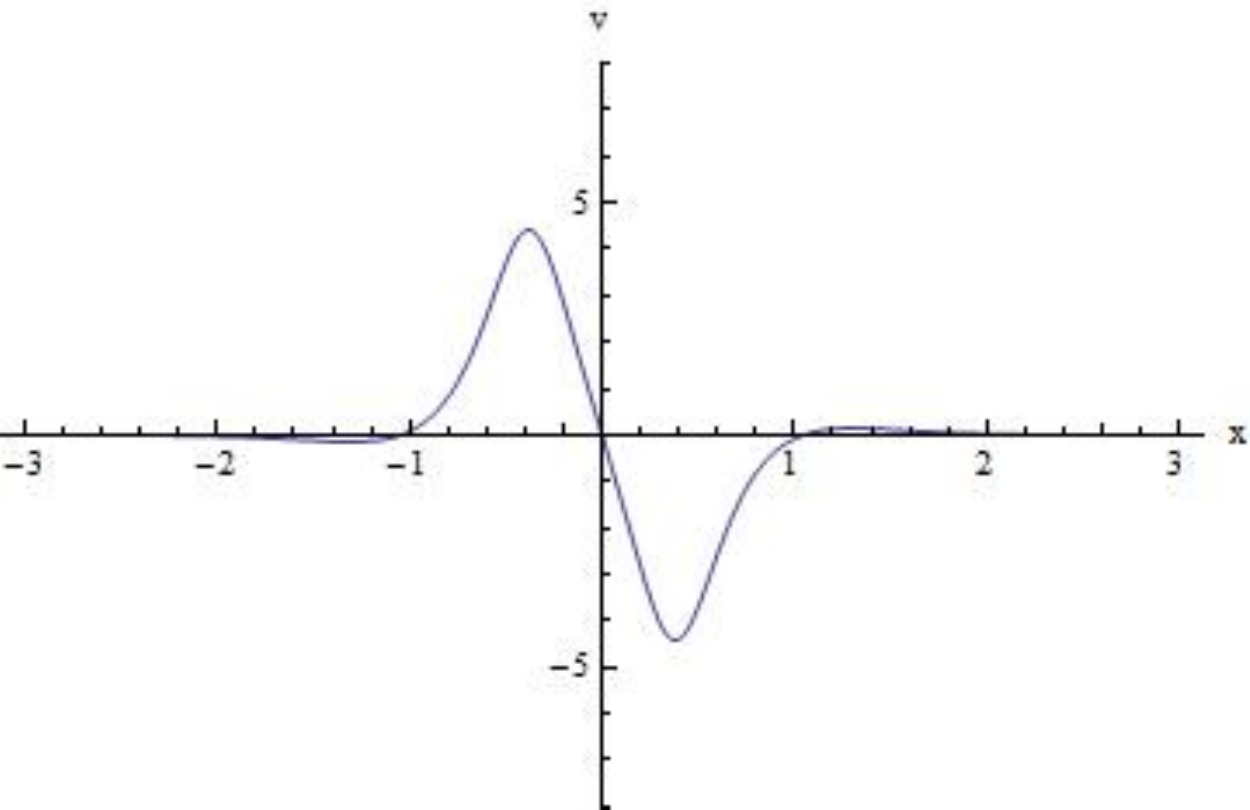}}}
\end{picture}
\end{center}
\vskip 1pt
\begin{minipage}{15cm}{\footnotesize
~~~~~~~~~~~~~~~~~~~~~(a)~~~~~~~~~~~~~~~~~~~~~~~~~~
~~~~~~~~~~~~~~~~~~(b)~~~~~~~~~~~~~~~~~~~~~~~~~~~~~~~~~~~~~~~~~~~~~(c)}
\end{minipage}
\caption{Breathers given by \eqref{solu-breather}.
(a). A moving breather with
$a_{1}=1+i,~b_{1}=2+0.5i$, $k_{1}=0.8-0.6i$.
(b). A stationary breather with
$a_{1}=b_{1}=1$, $k_{1}=\sqrt{3}+i$.
(c). A 2D-plot of (b) at $t=0$.
\label{Fig-5-1}}
\end{figure}

Fig.\ref{Fig-5-1}(a) shows an oscillating wave moving along a straight line.
The oscillation comes from the sine function and the frequency depends on both $x$ and $t$.
To understand more on the wave we use the coordinates
\begin{equation}
(Z=x+4t(3k_{12}^2-k_{11}^2),\, t)
\label{Z-t}
\end{equation}
to rewrite the solution \eqref{solu-breather} as
\begin{equation}
v=-2\Big(\arctan
\frac{2k_{11}\alpha\sin\big(2k_{12}(Z-8t(k_{11}^2+k_{12}^2))-\theta\big)}
{k_{12}
\big((a_{11}^2+a_{12}^2)e^{2k_{11}Z}
+(b_{11}^2+b_{12}^2)e^{-2k_{11}Z}
\big)}\Big)_Z,
\label{solu-breather-Z}
\end{equation}
and then fix $Z=0$, i.e., looking at the wave along the straight line $Z=0$.
Then it is clear that
\begin{itemize}
\item{The breather travels along the straight line $Z=0$,
in other words, the wave speed is $4(3k_{12}^2-k_{11}^2)$,
which means it admits bi-direction traveling.}
\item{A stationary breather appears when $3k_{12}^2=k_{11}^2$, as depicted in Fig.\ref{Fig-5-1}(b).}
\item{Under the coordinate system \eqref{Z-t} the frequency only depends on
$t$ and the period reads
\begin{equation}
T=\frac{\pi}{8k_{12}(k_{11}^2+k_{12}^2)}.
\end{equation}
}
\end{itemize}

We note that in two dimensions (fixing $t$) the breather is in fact a spindle-shape wave.
Let us go back to the solution \eqref{solu-breather}.
If we fix time $t$, then the breather oscillates with frequency
$\frac{|k_{12}|}{\pi}$ and its amplitude decays by the rate $e^{-2|x k_{11}|}$
as $x\to \pm\infty$. That means if $|k_{12}|$ is small and  $|k_{11}|$ is relatively large so that
the amplitude decay becomes the dominating factor,
we get a `normal' breather as shown in
Fig.\ref{Fig-5}(a); while if $|k_{11}|$ is small enough and
$|k_{12}|>>|k_{11}|$ so that the oscillation dominates, we will see a spindle-like wave
shown in Fig.\ref{Fig-5}(b).
In the latter case, the wave will travel with high speed and high oscillating frequency.
This can lead to overlaps of `normal' oscillating waves (like Fig.\ref{Fig-5}(a)) during their traveling,
which makes a spindle shape.

\begin{figure}[!h]
\setlength{\unitlength}{1mm}
\begin{center}
\begin{picture}(00,40)
\put(-65,-5){\resizebox{!}{3.7cm}{\includegraphics{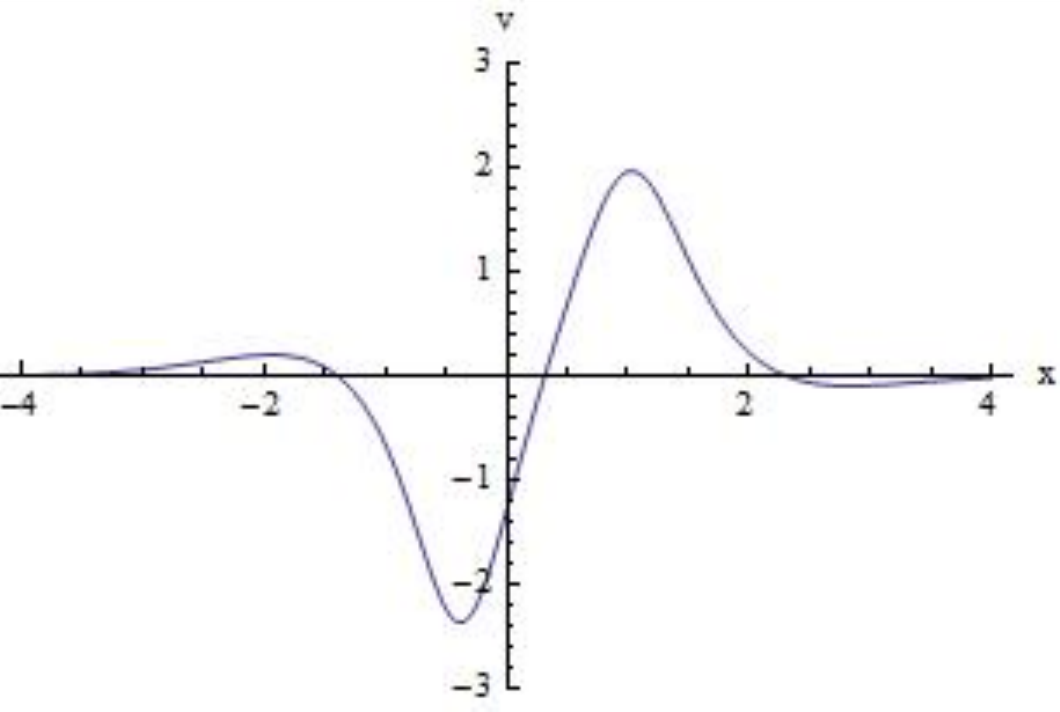}}}
\put(5,-5){\resizebox{!}{3.7cm}{\includegraphics{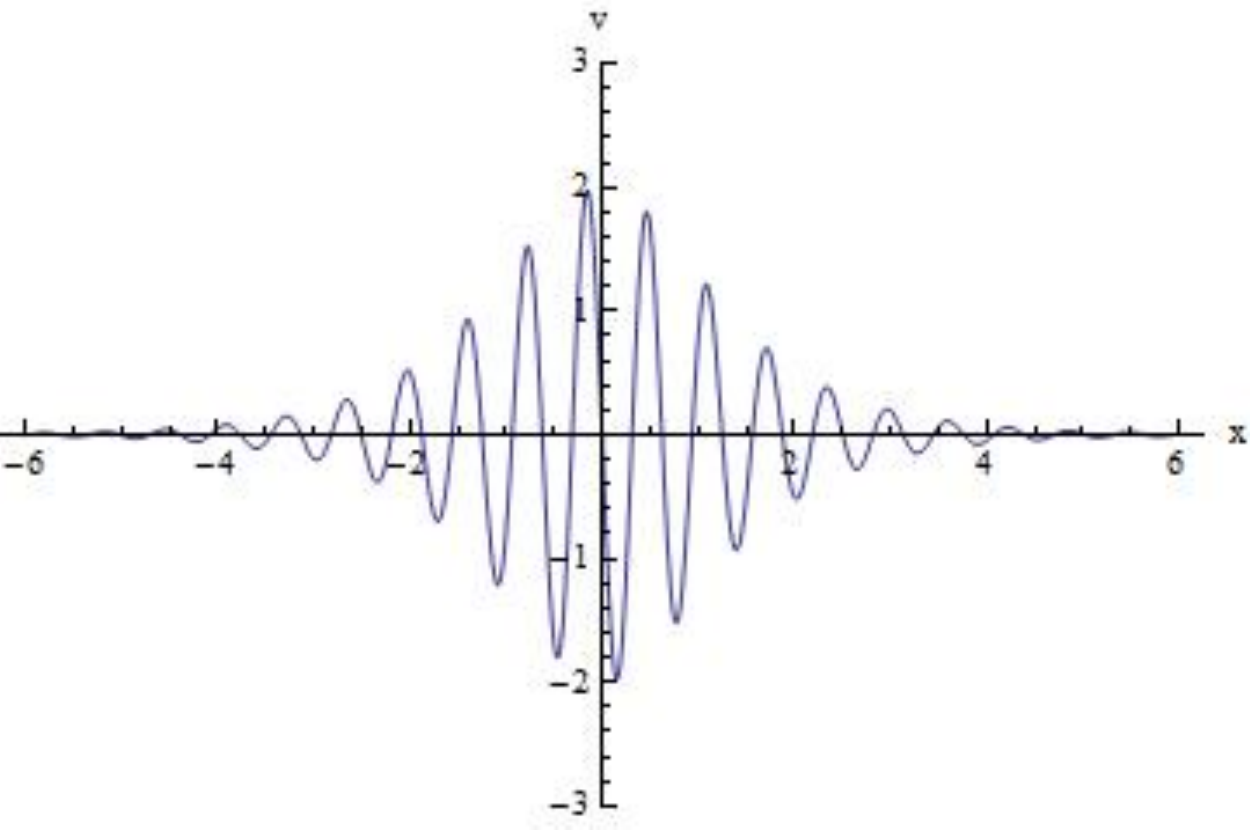}}}
\end{picture}
\end{center}
\vskip 1pt
\begin{minipage}{15cm}{\footnotesize
~~~~~~~~~~~~~~~~~~~~~~~~~~~~~~~~~~~~~(a)~~~~~~~~~~~~~~~~~~~~~~~~~~~~~~~~~~~
~~~~~~~~~~~~~~~~~~~~~~~~~(b)}
\end{minipage}
\caption{Shape of breathers given by \eqref{solu-breather}.
(a).  `Normal' shape at $t=0$ with $a_{1}=1+i,~b_{1}=2+0.5i$, $k_{1}=0.8-0.6i$.
(b). Spindle-shape breather at $t=0$ with
$a_{1}=1+i,~b_{1}=1+i$, $k_{1}=0.5+5i$.
\label{Fig-5}}
\end{figure}

Finally, in this subsection we list and depict two-breather solution and the simplest limit breather solution,
without further asymptotic analysis.
Both solutions can be given by \eqref{solu-mkdv}
with
$f$ being a 4 by 4 Wronskian
\begin{equation}
f=|\varphi,\varphi_{x},\varphi_{xx},\varphi_{xxx}|,
\label{solu-2-breather}
\end{equation}
where for the 2-breather solution
\begin{subequations}
\label{phi-2-breather}
\begin{align}
{\varphi}&=(\varphi_{11}, \varphi_{12}, \varphi_{21}, \varphi_{22})^{T},\\
\varphi_{j1}&= a_{j}  e^{\xi_j}+b_{j} e^{-{\xi}_j}, ~~
\varphi_{j2}=\bar{a}_{j}  e^{\bar{\xi}_j} -\bar{b}_{j}
e^{-{\bar{\xi}_j}}, ~~~
\xi_j=k_{j}x- 4 k_j^3 t+ \xi_{j}^{(0)},~~
\label{xi-2-breather}
\end{align}
\end{subequations}
in which $a_{j}, b_{j},k_j,\xi_{j}^{(0)}  \in \mathbb{C},$ and for the limit breather
\begin{equation}
{\varphi}=(\varphi_{11}, \varphi_{12}, \partial_{k_1}\varphi_{11}, \partial_{\bar{k}_1}\varphi_{12})^{T},
\label{phi-limit-breather}
\end{equation}
in which $\varphi_{11}, \varphi_{12}$ are defined by \eqref{xi-2-breather}.

Fig.\ref{Fig-6-1} shows the two-breather interaction where from the density plot (b) one can clearly see that
the two breathers are traveling along straight lines and a phase shift appears after interaction.
Fig.\ref{Fig-6} shows the shape and motion of a limit breather solution, where
from the density plot (b) one can clearly see that
the breather trajectories are not any longer straight lines,
Here we conjecture that they are governed by logarithm functions.

\begin{figure}[!h]
\setlength{\unitlength}{1mm}
\begin{center}
\begin{picture}(00,40)
\put(-65,-5){\resizebox{!}{4.0cm}{\includegraphics{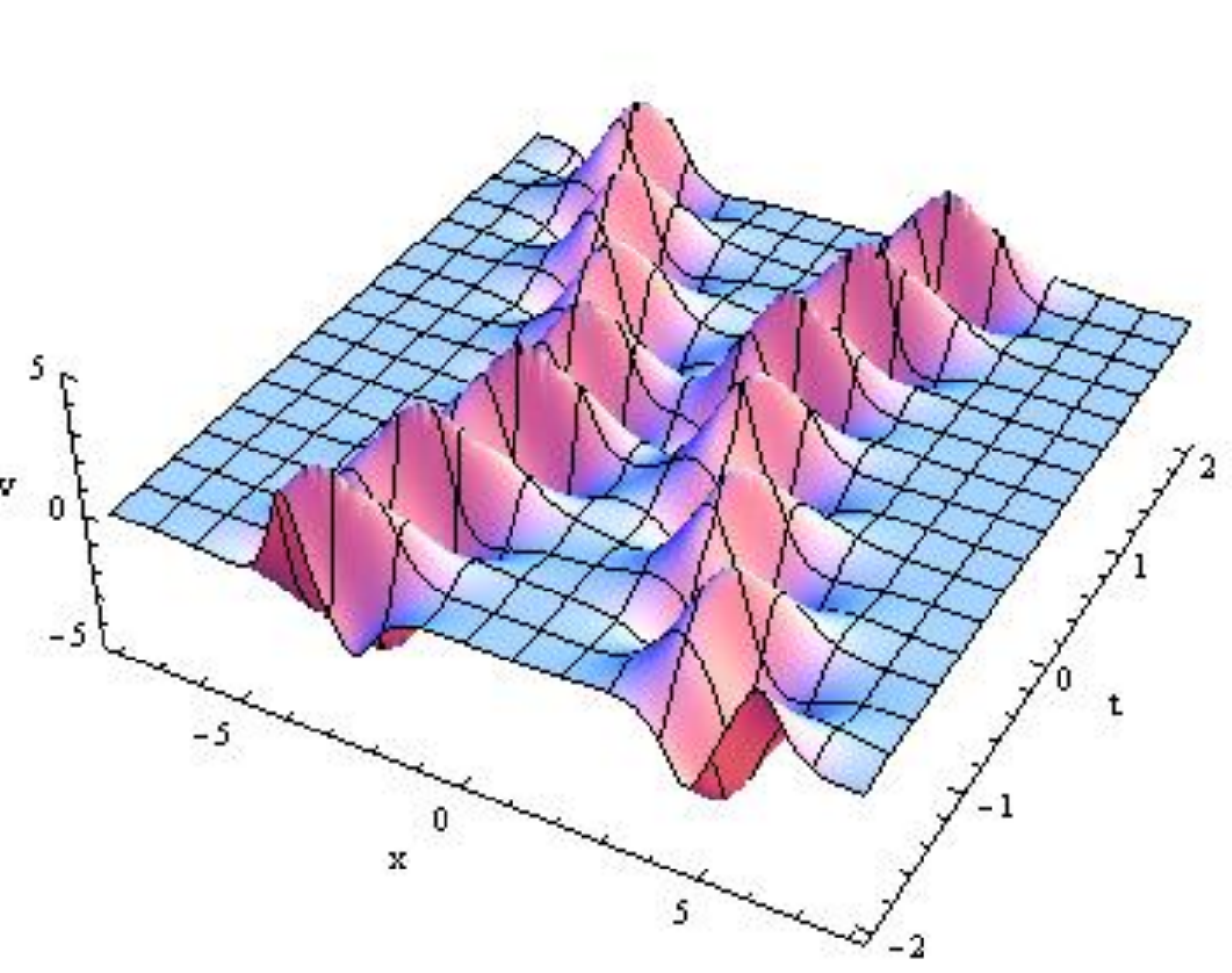}}}
\put(5,-5){\resizebox{!}{3.7cm}{\includegraphics{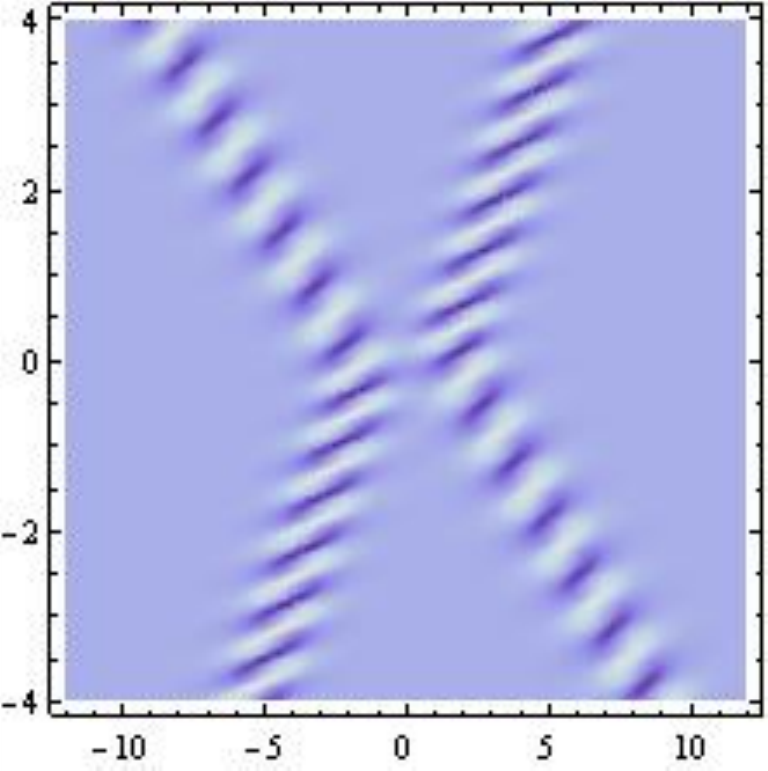}}}
\end{picture}
\end{center}
\vskip 1pt
\begin{minipage}{15cm}{\footnotesize
~~~~~~~~~~~~~~~~~~~~~~~~~~~~~~~~~~~~(a)~~~~~~~~~~~~~~~~~~~~~~~~~~~~~~
~~~~~~~~~~~~~~~~~~~~~~~~(b)}
\end{minipage}
\caption{Shape and motion of two breather solution given by \eqref{solu-mkdv} with \eqref{solu-2-breather} and \eqref{phi-2-breather}.
(a). 3D-plot for $a_1=b_1=a_2=b_2=1,~k_1=1+0.5i,~
k_2=0.8-0.6i$ and $\xi_1^{(0)}=\xi_2^{(0)}=0$.
(b). Density plot of (a)  for $x\in[-12,12],~ t\in [-4,4]$. \label{Fig-6-1}}
\end{figure}

\begin{figure}[!h]
\setlength{\unitlength}{1mm}
\begin{center}
\begin{picture}(00,40)
\put(-65,-5){\resizebox{!}{4.0cm}{\includegraphics{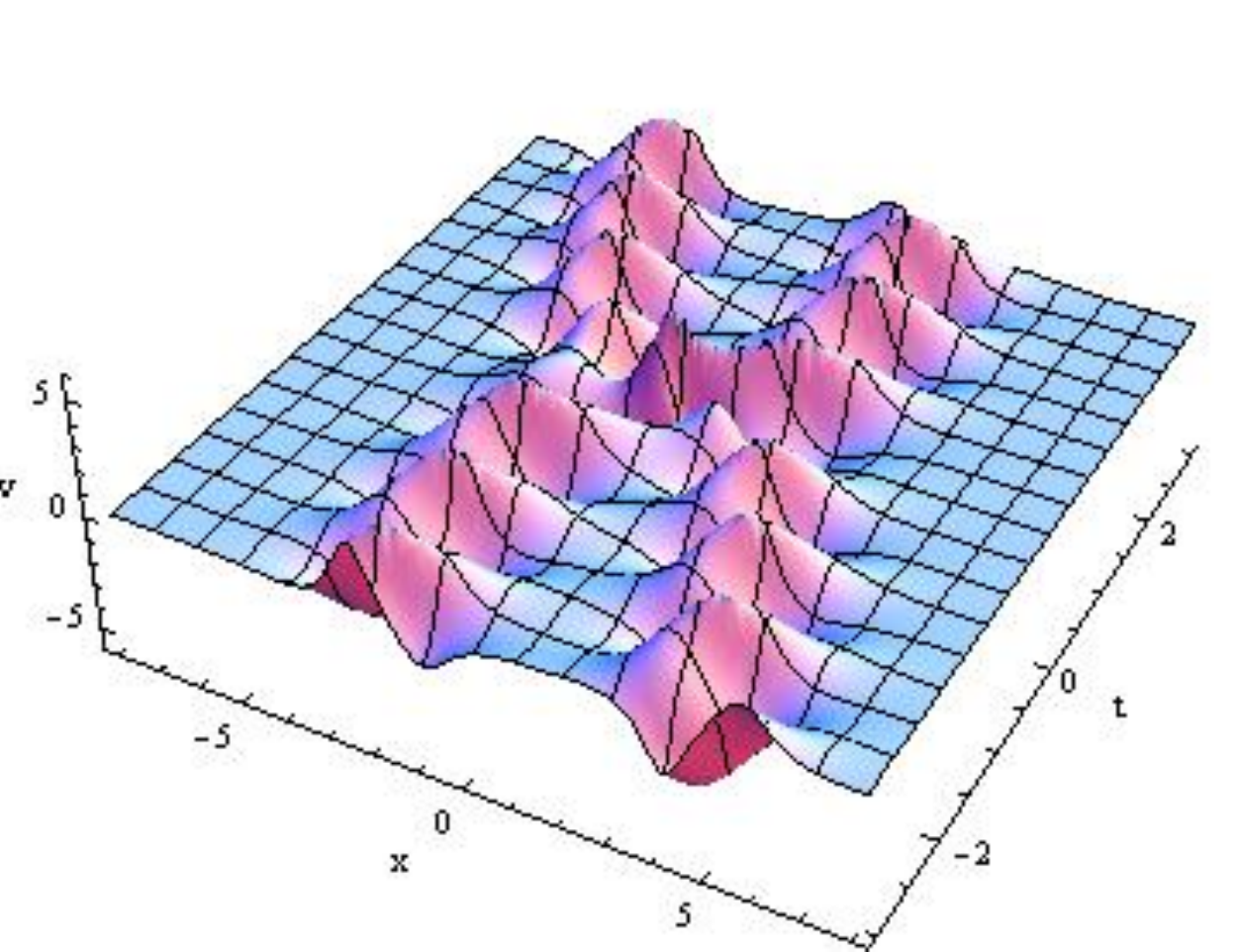}}}
\put(5,-5){\resizebox{!}{3.7cm}{\includegraphics{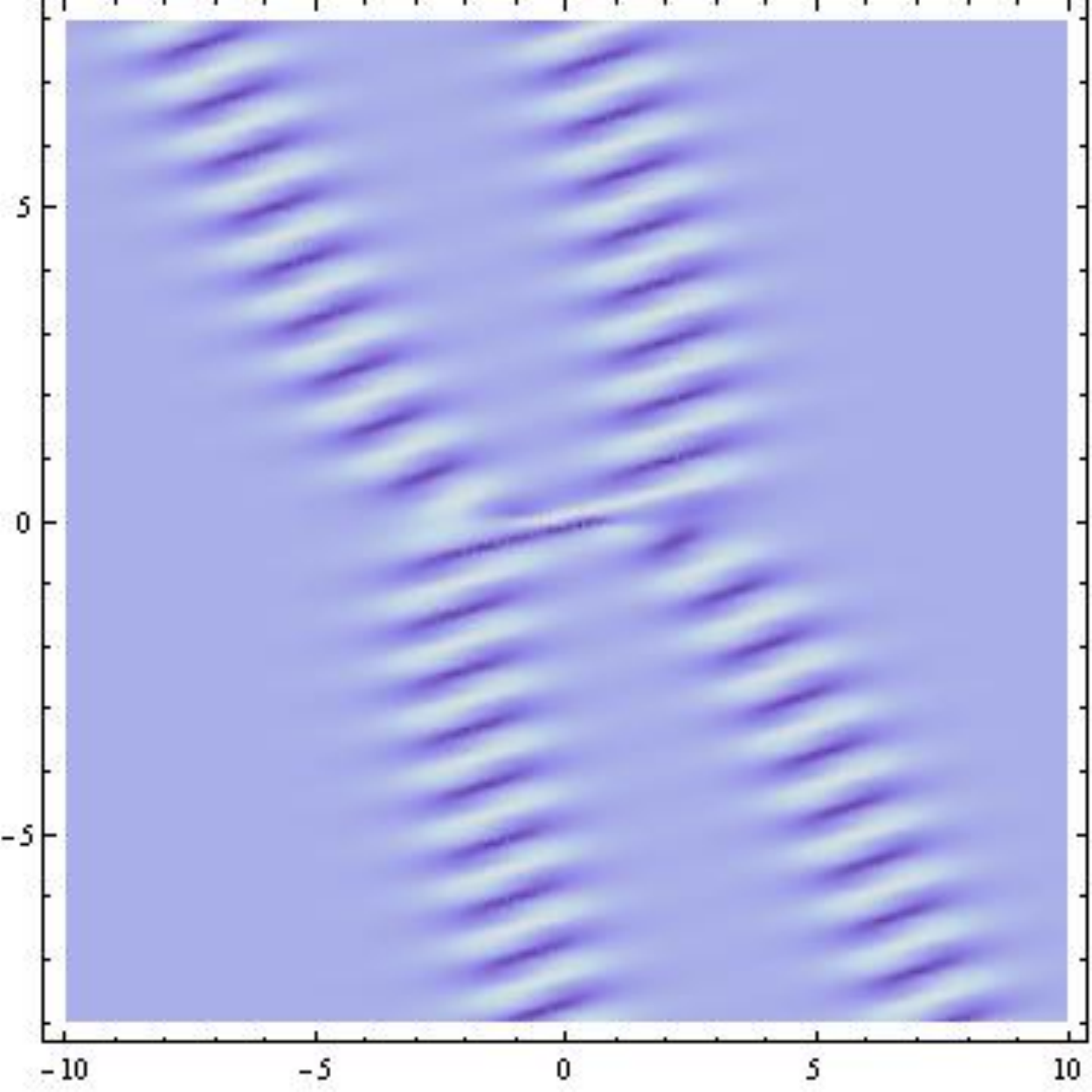}}}
\end{picture}
\end{center}
\vskip 1pt
\begin{minipage}{15cm}{\footnotesize
~~~~~~~~~~~~~~~~~~~~~~~~~~~~~~~~~~(a)~~~~~~~~~~~~~~~~~~~~~~
~~~~~~~~~~~~~~~~~~~~~~~~~~~~~~~~~~(b)}
\end{minipage}
\caption{Shape and motion of the limit breather solution given by \eqref{solu-mkdv} with \eqref{solu-2-breather} and \eqref{phi-limit-breather}.
(a). 3D-plot for  $a_1=b_1=1,~k_1=0.8+0.5i$  and $\xi_1^{(0)}=0$.
(b). Density plot of (a)  for $x\in[-10,10],~ t\in [-8,8]$. \label{Fig-6}}
\end{figure}

\subsection{Dynamics of rational solutions}

The first non-trivial rational solution  to the mKdV equation is \eqref{1rs}, i.e.,
\begin{equation}
v =v_{0}-\frac{4v_{0}}{4v_{0}^{2}X^{2}+1},~~X=x-6v_{0}^{2}t.
\label{1rs-dy}
\end{equation}
This is a non-singular traveling wave moving  with the constant
speed $6v_{0}^{2}$, constant amplitude $-3v_0$ and asymptotic line
$v=v_0$. It is depicted in Fig.\ref{Fig-8}.

\begin{figure}[!h]
\setlength{\unitlength}{1mm}
\begin{center}
\begin{picture}(00,40)
\put(-76,-5){\resizebox{!}{4.0cm}{\includegraphics{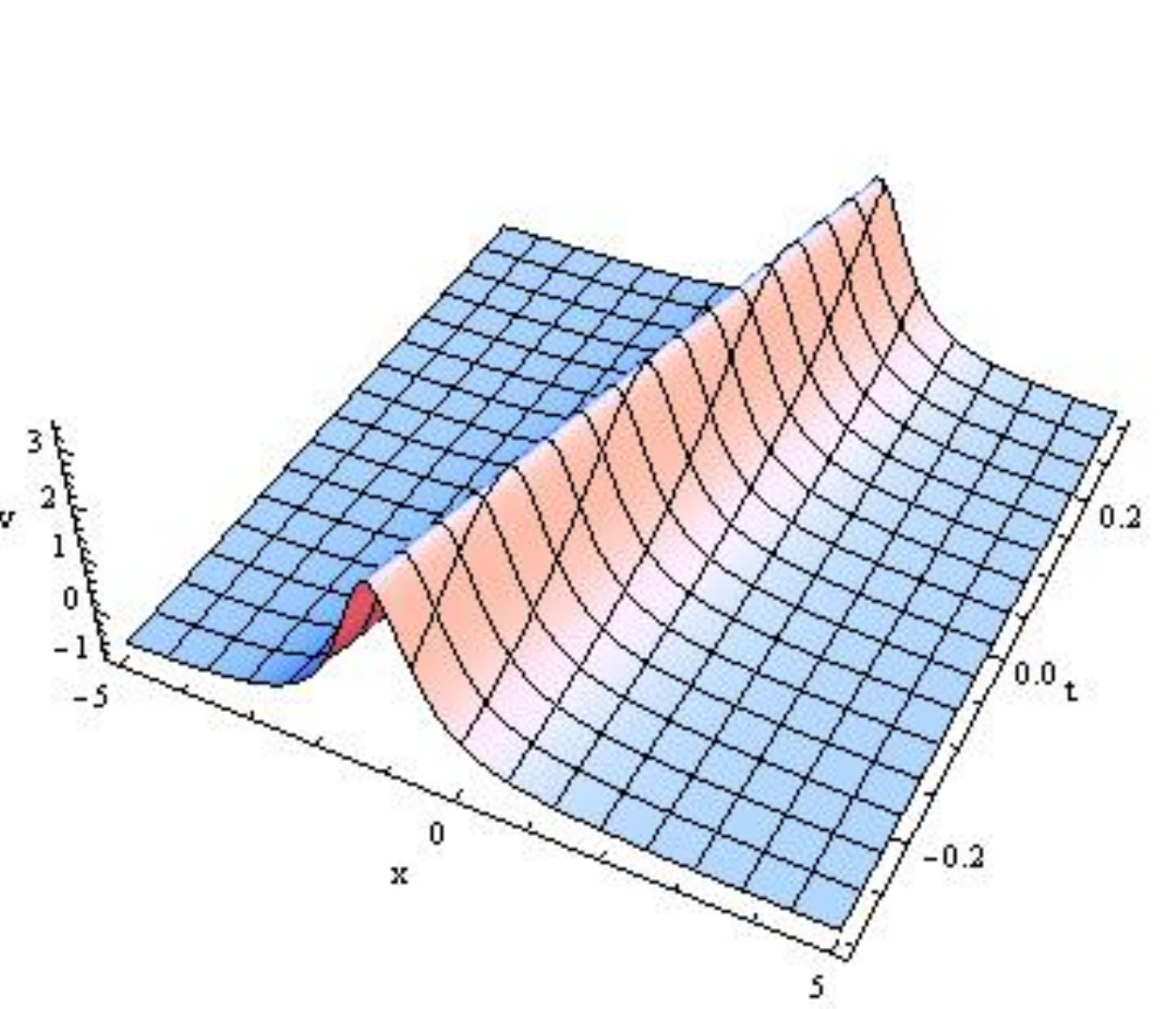}}}
\put(-29,-5){\resizebox{!}{3.5cm}{\includegraphics{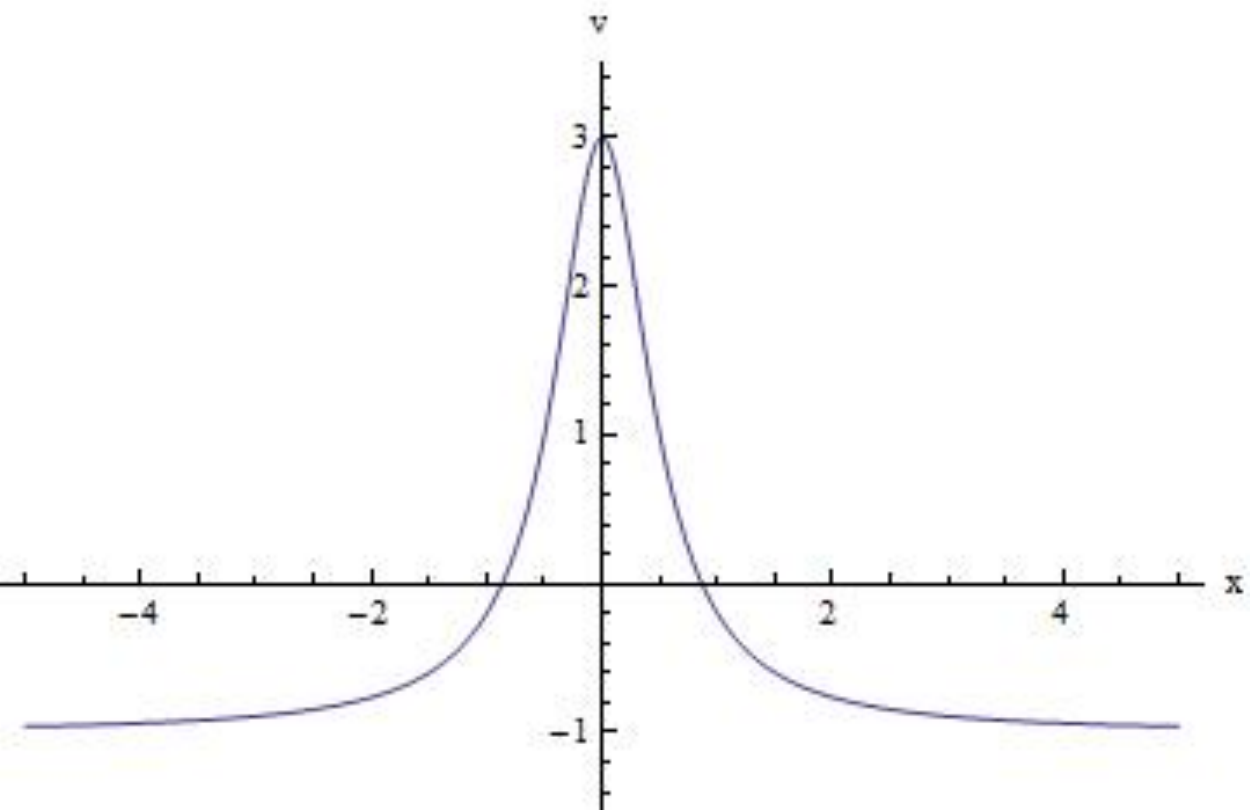}}}
\put(25,-5){\resizebox{!}{3.5cm}{\includegraphics{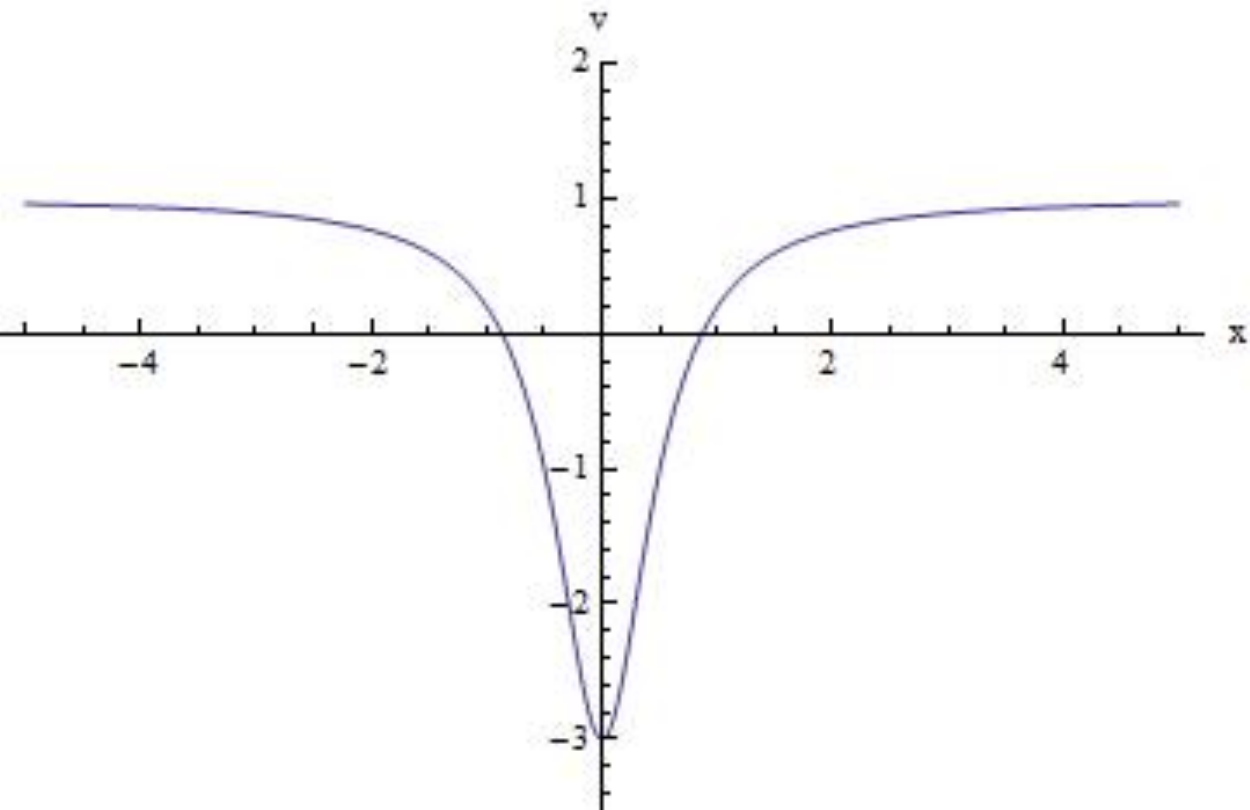}}}
\end{picture}
\end{center}
\vskip 1pt
\begin{minipage}{15cm}{\footnotesize
~~~~~~~~~~~~~~~~~~~~(a)~~~~~~~~~~~~~~~~~~~~~~~
~~~~~~~~~~~~~~~~~~~~~(b)~~~~~~~~~~~~~~~~~~
~~~~~~~~~~~~~~~~~~~~~~~~~~~~(c)}
\end{minipage}
\caption{Shape and motion of the rational solution given by
\eqref{1rs-dy} for $v_0=-0.8$ in (a), $v_0=-1$ in (b) and $v_0=1$ in
(c).} \label{Fig-8}
\end{figure}

The next rational solution is given by \eqref{2rs}, i.e.,
\begin{equation}
v
=v_{0}-\frac{12v_{0}(X^{4}+\frac{3}{2v_{0}^{2}}X^{2}-\frac{3}{16v_{0}^{4}}
-24Xt)}{4v_{0}^{2}(X^{3}+12t-\frac{3X}{4v_{0}^{2}})^{2}+9(X^{2}+
\frac{1}{4v_{0}^{2}})^{2}}, ~~X=x-6v_{0}^{2}t.
\label{2rs-dy}
\end{equation}
It can be  viewed as a double-traveling wave solution
\begin{equation}
v =v_{0}- \frac{12v_{0}(X^{4}+\frac{3}{2v_{0}^{2}}
XY-\frac{3}{16v_{0}^{4}})}
{4v_{0}^{2}(X^{3}-\frac{3Y}{4v_{0}^{2}})^{2}+9(X^{2}+\frac{1}{4v_{0}^{2}})^{2}},\label{2rs-XY}
\end{equation}
with
\begin{equation}
X=x-6v_{0}^{2}t, ~~Y=x-22v_{0}^{2}t.
\end{equation}
However, it does not show interactions of two single rational solutions.
Only one wave is left for large $x,t$ (see Fig.\ref{Fig-9}).
We re-depict Fig.\ref{Fig-9}(a) in Fig.\ref{Fig-10} by a density plot
so that we can see the wave top trace clearer.

\begin{figure}[!h]
\setlength{\unitlength}{1mm}
\begin{center}
\begin{picture}(00,40)
\put(-65,-5){\resizebox{!}{4.0cm}{\includegraphics{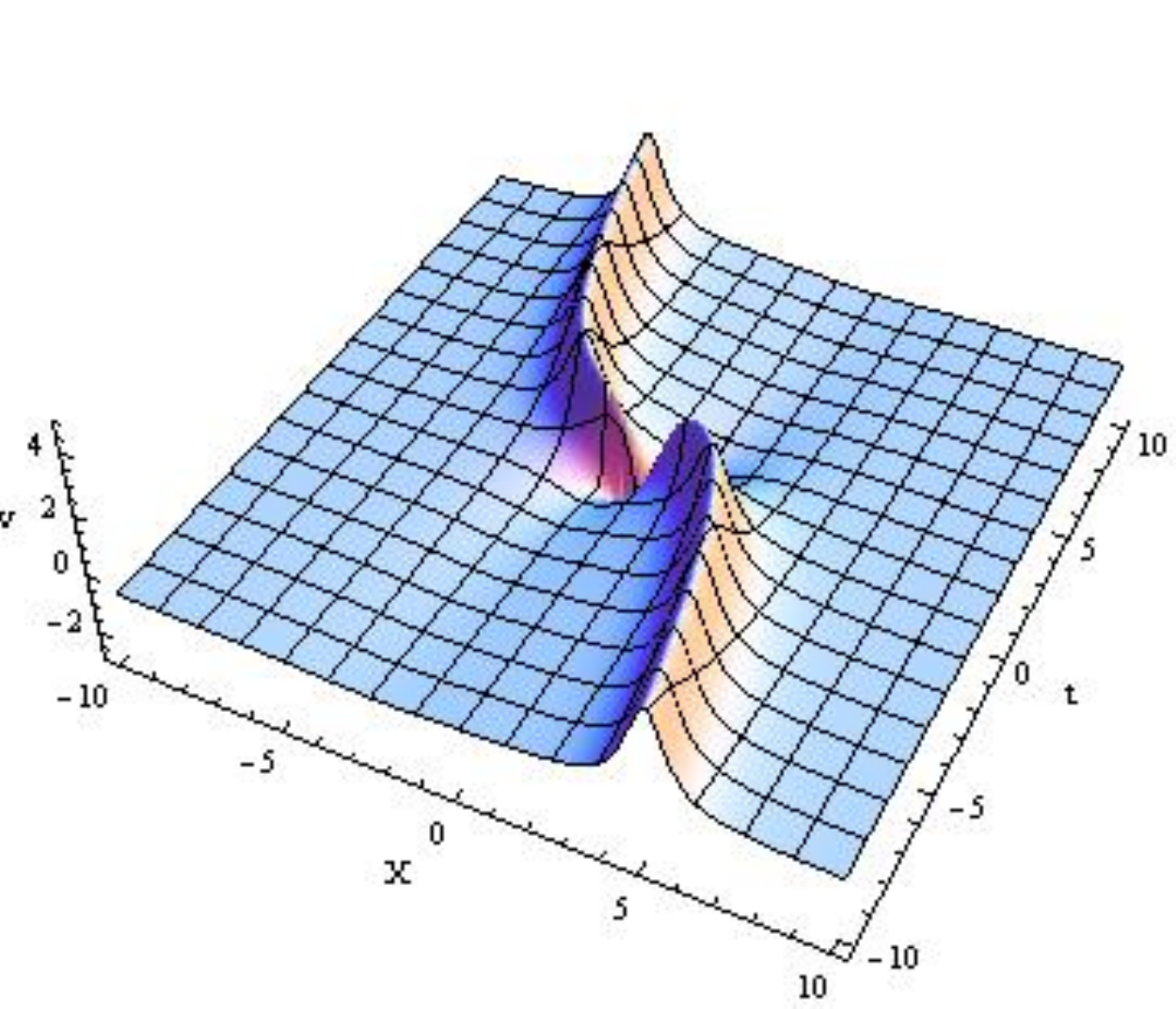}}}
\put(5,-5){\resizebox{!}{4.0cm}{\includegraphics{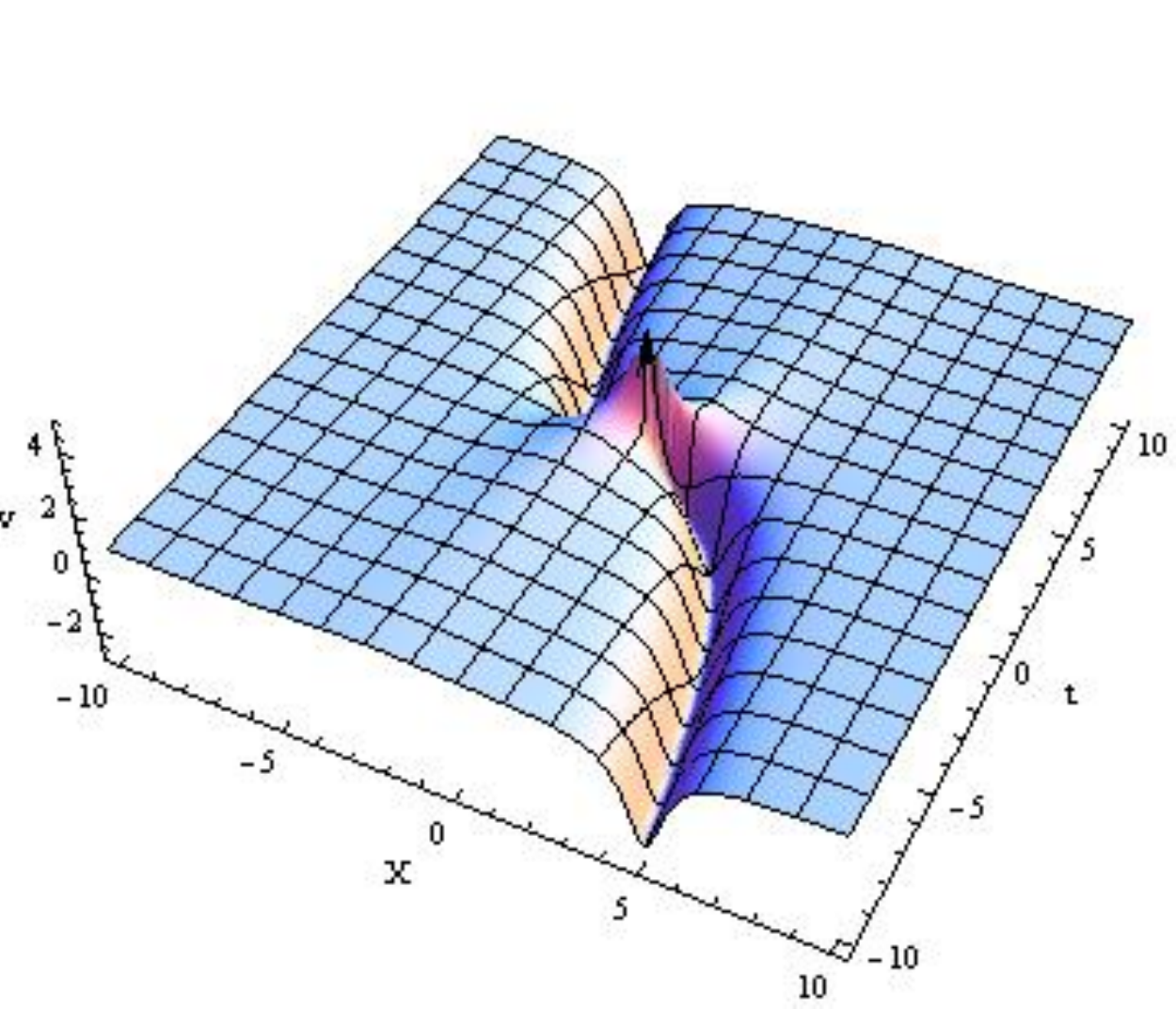}}}
\end{picture}
\end{center}
\vskip 1pt
\begin{minipage}{15cm}{\footnotesize
~~~~~~~~~~~~~~~~~~~~~~~~~~~~~~~~~~(a)~~~~~~~~~~~~~~~~~~~~~~~~~~~~~~~~~~~~~
~~~~~~~~~~~~~~~~~~~~~~~~(b)}
\end{minipage}
\caption{Shape and motion of the rational solution given by \eqref{2rs-dy} for
$v_0=-0.8$ in (a) and $v_0=0.8$ in (b).} \label{Fig-9}
\end{figure}

\begin{figure}[!h]
\setlength{\unitlength}{1mm}
\begin{center}
\begin{picture}(00,40)
\put(-64,-5){\resizebox{!}{4.0cm}{\includegraphics{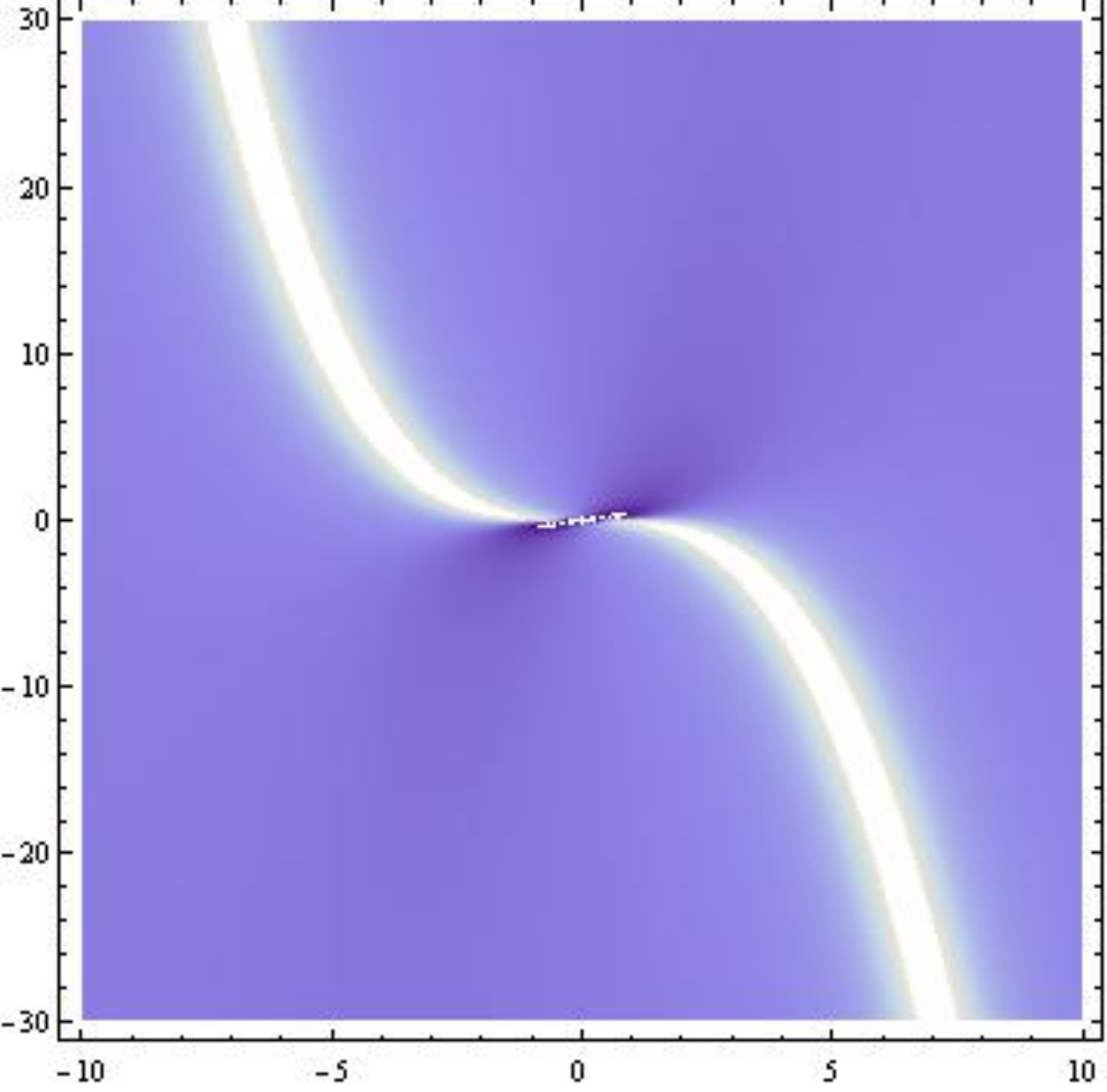}}}
\put(6,-5){\resizebox{!}{4.0cm}{\includegraphics{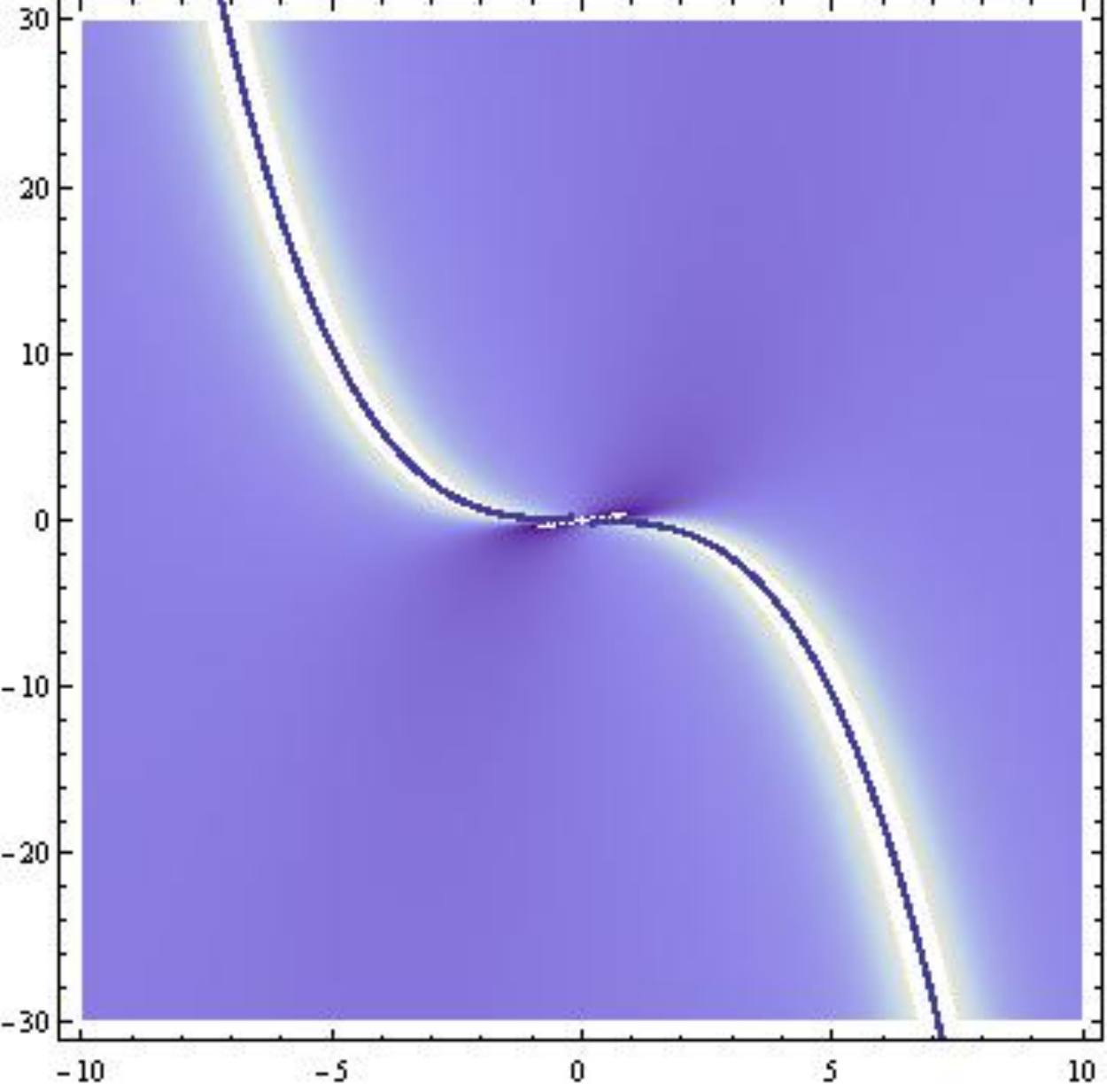}}}
\end{picture}
\end{center}
\vskip 1pt
\begin{minipage}{15cm}{\footnotesize
~~~~~~~~~~~~~~~~~~~~~~~~~~~~~~~~(a)~~~~~~~~~~~~~~~~~~~~~~~~~~~~~~~~~~~~
~~~~~~~~~~~~~~~~~~~~~~~~(b)}
\end{minipage}
\caption{(a): Density plot of Fig.\ref{Fig-9} (a) with $X\in[-10,10],~ t\in [-30,30]$.
(b): (a) overlapped by the trajectory curve given by \eqref{tract}.}
\label{Fig-10}
\end{figure}

To realize the asymptotic behavior analytically, we rewrite the solution \eqref{2rs-dy}
in the following  coordinates system
\begin{equation}
\big(X,~T=X^3+12t+\frac{3}{16v_{0}^{4}X}\big)
\end{equation}
and this gives
\begin{equation}
v =v_{0}
-\frac{12v_{0}(3X^{4}+\frac{3}{2v_{0}^{2}}X^{2}-2TX+\frac{3}{16v_{0}^{4}})}
{4v_{0}^{2}(T-\frac{3}{16v_{0}^{4}X}-\frac{3X}{4v_{0}^{2}})^{2}+9(X^{2}+\frac{1}{4v_{0}^{2}})^{2}}.
\label{2rs-XT}
\end{equation}
Then, by calculation it can be found that for given large $X$ the
wave \eqref{2rs-XT} has a single stationary point at $T=0$ where $v$
gets a local extreme value
\begin{equation}
v=\frac{v_0(1-12v_0^2 X^2)}{1+4 v_0^2 X^2},
\end{equation}
which  goes to $-3v_0$ as $X\to \pm\infty$.
Thus we can conclude that for large $X,t$ the wave asymptotically travels along the curve
\begin{equation}
T=X^3+12t+\frac{3}{16v_{0}^{4}X}=0
\label{tract}
\end{equation}
with amplitude $-3v_0$.
Fig.\ref{Fig-10}(b) displays a density plot overlapped by the above wave trajectory curve.

More details on the rational solutions to the mKdV equation can be
found in \cite{ZH-S-rational} where the rational solutions are
derived via bilinear B\"acklund transformation.

\section{Conclusions}

\subsection{Summation}
In the paper we reviewed the Wronskian solutions to the mKdV equation \eqref{mKdV} in terms of Wronskians.
When a solution is expressed through the Wronskian
\begin{equation}
f=f(\varphi)=|\widehat{N-1}|,
\label{nss-w-1}
\end{equation}
one needs to solve the finalized CES \eqref{mKdV-condition} together with \eqref{cc},
i.e.,
\begin{subequations}
\label{mKdV-condition-1}
\begin{align}
\varphi_{xx}=&\mathbb{A}\varphi, \label{mKdV-condition-a-1}\\
\varphi_{x}=&\mathbb{B}\bar{\varphi}, \label{mKdV-condition-b-1}\\
\varphi_{t}=&-4\varphi_{xxx},\label{mKdV-condition-c-1}
\end{align}
\end{subequations}
and
\begin{equation}
\mathbb{A}=\mathbb{B}\bar{\mathbb{B}}. \label{cc-1}
\end{equation}
$\mathbb{A}$ is the auxiliary matrix that we introduced to deal with
the complex operation in \eqref{mKdV-condition-b-1} and it works in
practice. As a result, with the help of $\mathbb{A}$ we solved the
above CES and then categorized the solutions to the mKdV equation
in terms of the canonical form of $\mathbb{A}$ (rather than canonical
form of $\mathbb{B}$). Solutions are categorized by solitons
(together with their limit case) and breathers (together with their
limit case). There are no rational solutions arising from
\eqref{mKdV-condition-1} because no rational solutions correspond to
zero eigenvalues of $\mathbb{A}$ while we need $|\mathbb{B}|\neq 0$
to finish Wronskian verification. To derive rational solutions for
the mKdV equation \eqref{mKdV}, we employed the Galilean transformed
equation, i.e., the KdV-mKdV equation \eqref{kdv-mkdv} which admits
rational solutions in Wronskian form. Then the rational solutions to
the mKdV equation can be recovered through the inverse
transformation. Dynamics of some obtained solutions was analyzed
and illustrated. Here, particularly, we would like to sum up a
typical characteristic of limit solitons: the wave trajectories
asymptotically follow  logarithm curves (combined with  linear
functions). This point is based on several examples we have
examined\cite{ZDJ-lim-rev}.

Obviously, through the Galilean transformation \eqref{GT},
all these obtained solutions of the  mKdV equation \eqref{mKdV} can easily be used for
the KdV-mKdV equation \eqref{kdv-mkdv} which often appears in physics contexts.
In fact, in the paper we do not differ them from each other.
In addition to the  KdV-mKdV equation,
our treatment to the complex operation in \eqref{mKdV-condition-1}
can also be applied to the sine-Gordon equation.

There are Miura transformations between the KdV equation and the mKdV equation (for both $\varepsilon=\pm1$).
For the mKdV($\varepsilon=-1$) equation the Miura transformation
provides a real map between solutions of the KdV equation and the mKdV equation($\varepsilon=-1$), c.f.\cite{Gesztesy-Tams-1991}.
However, when $\varepsilon=1$, i.e., for the mKdV equation \eqref{mKdV}, the Miura transformation
has to be complex. In more detail,
it maps  the real mKdV equation \eqref{mKdV} to a complex KdV equation.
That means we have had a nice determinant expression  for the complex KdV equation.
Further investigation about this will be considered separately.

\subsection{List of solutions}
Let us list out the obtained solutions and their corresponding basic Wronskian vectors.
Solutions to the mKdV equation \eqref{mKdV} can be given by
\begin{subequations}
\label{nss}
\begin{equation}
v=2\Big(\mathrm{arctan}\frac{F_2}{F_1}\Big)_{x}
=\frac{-2(F_{1,x}F_2-F_1F_{2,x})}{F_2^2+F_1^2},
\label{nss-s}
\end{equation}
 where
\begin{equation}
f=f(\varphi)=|\widehat{N-1}|=F_1+iF_2,~~F_1= \mathrm{Re}[f],~~F_2=\mathrm{Im}[f].
\label{nss-w}
\end{equation}
\end{subequations}
The available    Wronskian vectors are the following.
\begin{itemize}
\item{
\textbf{For soliton solutions:}
\begin{subequations}
\label{nss-phi-soliton}
\begin{equation}
\varphi=\varphi^{[s]}_{N}=(\varphi_{1},  \varphi_{2}, \cdots,  \varphi_{N})^{T},
\end{equation}
with
\begin{equation}
\varphi^{}_{j}= a_{j}^+  e^{\xi_{j}}+ i a_{j}^-
 e^{-\xi_{j}}, ~\xi_{j}=k_{j}x-4k_{j}^{3}t+\xi_{j}^{(0)},~
 a_{j}^+ , a_{j}^-, k_j, \xi_{j}^{(0)} \in \mathbb{R}.
\label{nss-phi-soliton-j}
\end{equation}
\end{subequations}
}
\item{\textbf{For limit solutions of solitons:}
\begin{subequations}
\label{nss-phi-soliton-lim}
\begin{equation}
{\varphi}=\varphi^{[ls]}_{N}(k_1)=\mathcal{A^+} \mathcal{Q}_{0}^{+}+ i \mathcal{A^-}
\mathcal{Q}_{0}^{-},~~\mathcal{A}^{\pm} \in \widetilde{G}_N(\mathbb{R}),
\end{equation}
with
\begin{equation}
 \mathcal{Q}^{\pm}_0=(\mathcal{Q}^{\pm}_{0, 0},
\mathcal{Q}^{\pm}_{0,1}, \cdots, \mathcal{Q}^{\pm}_{0, N-1})^T,~~
\mathcal{Q}^{\pm}_{0, s}=\frac{1}{s!}\partial^{s}_{k_1}e^{\pm
\xi_1},
\end{equation}
\end{subequations}
where $\xi_1$ is defined in \eqref{nss-phi-soliton-j}.
}
\item{
\textbf{For breather solutions:}
\begin{subequations}
\label{nss-phi-breather}
\begin{equation}
{\varphi}=\varphi^{[b]}_{2N}=(\varphi_{11}, \varphi_{12}, \varphi_{21}, \varphi_{22},
\cdots, \varphi_{N1}, \varphi_{N2})^{T},
\end{equation}
with
\begin{eqnarray}
\varphi_{j1}=&& a_{j}  e^{\xi_j}+b_{j} e^{-{\xi}_j}, ~~
\varphi_{j2}=\bar{a}_{j}  e^{\bar{\xi}_j} -\bar{b}_{j}
e^{-{\bar{\xi}_j}}, \\
\xi_j=&& k_{j}x- 4 k_j^3 t+ \xi_{j}^{(0)},~~a_{j}, b_{j},\xi_{j}^{(0)}  \in \mathbb{C}.
\label{nss-phi-breather-j}
\end{eqnarray}
\end{subequations}
}
\item{
\textbf{For limit solutions of breathers:}
\begin{subequations}
\label{nss-phi-breather-lim}
\begin{equation}
\varphi=\varphi^{[lb]}_{2N}(k_1)=(\varphi^+_{1, 1}, \varphi^-_{1,
2}, \varphi^+_{2, 1}, \varphi^-_{2, 2}, \cdots, \varphi^+_{N, 1},
\varphi^-_{N, 2})^{T},
\end{equation}
and the elements are given through
\begin{align}
\varphi^{+}&=(\varphi^+_{1, 1}, \varphi^+_{2, 1}, \cdots,
\varphi^+_{N,1})^{T}
=\mathcal {A}\mathcal{Q}^{+}_{0}+\mathcal{B}\mathcal {Q}^{-}_{0},\\
\varphi^{-}&=(\varphi^-_{1, 2}, \varphi^-_{2, 2}, \cdots,
\varphi^-_{N, 2})^{T} =\bar{\mathcal
{A}}\bar{\mathcal{Q}}^{+}_{0}-\bar{\mathcal{B}}\bar{\mathcal{Q}}^{-}_{0},
\end{align}
where $\mathcal{A}, \mathcal{B} \in \widetilde{G}_N(\mathbb{C})$,
\begin{equation}
 \mathcal{Q}^{\pm}_0=(\mathcal{Q}^{\pm}_{0, 0},
\mathcal{Q}^{\pm}_{0,1}, \cdots, \mathcal{Q}^{\pm}_{0, N-1})^T,~~
\mathcal{Q}^{\pm}_{0, s}=\frac{1}{s!}\partial^{s}_{k_1}e^{\pm
\xi_1},
\end{equation}
\end{subequations}
and $\xi_1$ is defined in \eqref{nss-phi-breather-j}.
}
\end{itemize}

We note that, thanks to the linear property of the CES \eqref{mKdV-condition},
one may also get mixed solutions by arbitrarily combining the above  vectors to be a new
Wronskian vector. For example, take
\begin{equation}
\varphi=\left(\begin{array}{c}
\varphi^{[s]}_{N_1}\\
\varphi^{[ls]}_{N_2}(k_{N_1+1})
\end{array}
\right).
\end{equation}
The related solution corresponds to the interaction
between $N_1$-soliton and a $(N_2-1)$-order limit-soliton  solutions.

Finally, for the  rational solution,
it is given by
\begin{subequations}
\label{rss}
\begin{equation}
v(x,t)=v_{0}-\frac{2(F_{1,X}F_2-F_1 F_{2,X})}{F_2^2+F_1^2},~~
X=x-6v_{0}^{2}t,~~v_0\neq 0\in \mathbb{R}, \label{rss-s}
\end{equation}
where still
\begin{equation}
f=f(\psi)=|\widehat{N-1}|=F_1+iF_2,~~F_1= \mathrm{Re}[f],~~F_2=\mathrm{Im}[f],
\label{rss-w}
\end{equation}
\end{subequations}
and the Wronskian is composed by
\begin{subequations}
\begin{equation}
\psi =(\psi_1,\psi_2,\cdots,\psi_N)^T,
\end{equation}
with
\begin{equation}
\psi_{j+1}= \frac{1}{(2j)!}\frac{\partial ^{2j} }{{\partial
k_1}^{2j}}\varphi_{1}\,\bigg|_{k_1=0},~~(j=0,1,\cdots,N-1),
\end{equation}
and
\begin{equation}
\varphi_{1}=\sqrt{2v_{0}+2ik_{1}}\,e^{\eta_1}+\sqrt{2v_{0}-2ik_{1}}\,
e^{-\eta_1},~ \eta_1=k_{1}X-4k_{1}^{3}t,~k_1\in \mathbb{R} .
\end{equation}
\end{subequations}

\section*{Acknowledgments}

The authors sincerely thank Prof. Gesztesy for kindly providing Refs.\cite{Gesztesy-Tams-1991,Gesztesy-Rmp-1991}.
This project is supported by the NSF of China (No. 11071157),
Specialized Research Fund for the Doctoral Program of Higher Education of China
(No. 20113108110002),
Shanghai Leading Academic Discipline Project (No. J50101) and
Postgraduate Innovation Foundation of Shanghai University (No.
SHUCX111027).


\appendix

\section{Proof of Theorem \ref{Th 2.1}}
\label{A:sec-1}
\begin{proof}
The compatibility of \eqref{cond-a} and \eqref{cond-b}, i.e., $\phi_{xt}=\phi_{tx}$, yields \eqref{W_T}.
Using \eqref{cond-a} one gets the complex conjugate form of $f$ as
\begin{equation}
\bar{f}=|\bar{B}(t)||-1, \widehat{N-2}|.\label{eqA1}
\end{equation}
Then the necessary derivatives of $f$ and $\bar{f}$  are presented as the following,
\begin{subequations}
\begin{align}
f_x=&|\widehat{N-2}, N|,\\
f_{xx}=&|\widehat{N-3}, N-1, N|+|\widehat{N-2}, N+1|,  \\
f_{xxx}=&|\widehat{N-4}, N-2, N-1, N|+|\widehat{N-2}, N+2|\nonumber\\
&+2|\widehat{N-3}, N-1, N+1|,\\
f_t=&-4(|\widehat{N-4}, N-2, N-1, N|-|\widehat{N-3}, N-1,
N+1|\nonumber\\
&+|\widehat{N-2}, N+2|)+\mbox{tr}(C(t))|\widehat{N-1}|,\label{f-t}
\end{align}
\label{fde}
\end{subequations}
and
\begin{subequations}
\begin{align}
\bar{f}_{x}=&|\bar{B}(t)||-1,\widehat{N-3}, N-1|,   \\
\bar{f}_{xx}=&|\bar{B}(t)|(|-1,\widehat{N-4}, N-2, N-1|+|-1,\widehat{N-3}, N|),  \\
\bar{f}_{xxx}=&|\bar{B}(t)|(|-1,\widehat{N-3}, N+1|+2|-1,\widehat{N-4}, N-2,N|\nonumber\\
&+|-1,\widehat{N-5},N-3,N-2, N-1|),\\
\bar{f}_t=&-4|\bar{B}(t)|(|-1,\widehat{N-5},,N-3,N-2, N-1|\nonumber\\
&-|-1,\widehat{N-4}, N-2,N|+|-1,\widehat{N-3},
N+1|)\nonumber\\
&+|\bar{B}(t)|_t|-1,\widehat{N-2}|+\mbox{tr}(C(t))|\bar{B}(t)||-1,\widehat{N-2}|.\label{f-bar-t}
\end{align}
\label{barfde}
\end{subequations}
Using the condition \eqref{cond-a} the complex conjugate of \eqref{f-t} is
\begin{align*}
\bar{f}_t=&-4|\bar{B}(t)|(|-1,\widehat{N-5},N-3,N-2, N-1|-|-1,\widehat{N-4}, N-2,N|\nonumber\\
&+|-1,\widehat{N-3},
N+1|)+\mbox{tr}(\bar{C}(t))|\bar{B}(t)||-1,\widehat{N-2}|,\label{f-bar-t}
\end{align*}
which should be same as \eqref{f-bar-t}.
This requires $|B(t)|_t=0$ and $\mathrm{tr}(C(t))\in \mathbb{R}(t)$, i.e., the condition \eqref{B-qiudao}.

Noting that $\phi_{xx}=B(t)\bar{B}(t)\phi$ and using Proposition \ref{Prop 2.1} with
$\Omega_{j,s}=\partial_x^2$, we find
\begin{align}
& \mbox{tr}(B(t)\bar{B}(t))|-1,\widehat{N-3},
N-1|=-|-1,\widehat{N-5},N-3, N-2, N-1|\nonumber\\
&~~~~~~~~~~~~~~~~~~~~~~~~~~~~~~~~~~~~~~~~~~~~~
+|-1,\widehat{N-3}, N+1|,\\
& \mbox{tr}(B(t)\bar{B}(t))|\widehat{N-2},N|=-|\widehat{N-4},
N-2,N-1,N|+|\widehat{N-2}, N+2|,\\
& \mbox{tr}(B(t)\bar{B}(t))|-1,\widehat{N-2}|=-|-1,\widehat{N-4}, N-2, N-1|+|-1,\widehat{N-3}, N|,\\
& \mbox{tr}(B(t)\bar{B}(t))|\widehat{N-1}|=-|\widehat{N-3},
N-1,N|+|\widehat{N-2}, N+1|.
\label{equalities}
\end{align}
Then, substituting \eqref{fde} and \eqref{barfde} into \eqref{blinear-mKdV1} and making use of \eqref{equalities}, we have
\begin{eqnarray*}
&& ~~~\bar{f_{t}}f-\bar{f}f_t+\bar{f}_{xxx}f-3\bar{f}_{xx}f_x+3\bar{f_x}f_{xx}-\bar{f}f_{xxx}\nonumber\\
&&=6|\bar{B}(t)|\big(-|-1,\widehat{N-3}, N+1||\widehat{N-1}|-|-1,\widehat{N-4}, N-2, N-1||\widehat{N-2}, N|\nonumber\\
&&~~
+|\widehat{N-4}, N-2,  N-1,N||-1,\widehat{N-2}|-|\widehat{N-3}, N-1, N+1||-1,\widehat{N-2}|\nonumber\\
&&~~+|\widehat{N-2}, N+1||-1,\widehat{N-3}, N-1|-|-1,\widehat{N-4}, N-2, N||\widehat{N-1}|\big),
\end{eqnarray*}
which is zero in the light of Proposition \ref{Prop 2.2}. Similarly, one can prove \eqref{blinear-mKdV2}. Thus the proof is completed.
\end{proof}

\section{Eigen-polynomial  of $\mathbb{A}=\mathbb{B}\bar{\mathbb{B}}$}
\label{A:sec-2}

We prove the Proposition \ref{Prop 3.1} through the following 2
Lemmas.

\begin{lemma}\label{lem:A-1}
For two arbitrary $N$th-order complex matrices $A$ and $B$,
\begin{eqnarray}
\mathrm{det}(\lambda I_N-AB)=\mathrm{det}(\lambda I_N-BA),
\label{lem}
\end{eqnarray}
where $I_N$ is the $N$th-order unit matrix.
\end{lemma}
\begin{proof}
Assuming that $\mbox{rank}(A)=r$, then there exist $N$th-order non-singular matrices $P$ and $Q$ such that
\begin{equation}
A =P\left ( \begin{array}{ll}
                I_r & 0\\
                0   & 0
                \end{array}
        \right )Q.
\label{A-P-Q}
\end{equation}
Thus
\begin{eqnarray*}
\mbox{det}(\lambda I_N-AB)&=&\mbox{det}(\lambda I_N-P\left (
\begin{array}{ll}
                I_r & 0\\
                0   & 0
                \end{array}
        \right )QB)\nonumber\\
&=&\mbox{det}(P^{-1}(\lambda I_N-P\left (
\begin{array}{ll}
                I_r & 0\\
                0   & 0
                \end{array}
 \right )QB)P)\nonumber\\
&=&\mbox{det}(\lambda I_N-P^{-1}(P\left (
\begin{array}{ll}
                I_r & 0\\
                0   & 0
                \end{array}
 \right )QB)P)\nonumber\\
&=&\mbox{det}(\lambda I_N-\left (
\begin{array}{ll}
                I_r & 0\\
                0   & 0
                \end{array}
 \right )QBP)
        .\label{p-1}
\end{eqnarray*}
In a similar way we have
\begin{eqnarray}
&& \mbox{det}(\lambda I_N-BA)=\mbox{det}(\lambda I_N-QBP
\left (
\begin{array}{ll}
                I_r & 0\\
                0   & 0
                \end{array}
 \right )).
\label{p-2}
\end{eqnarray}
If we rewrite the matrix $QBP$ into the following block matrix form with same structure as
$\Big(
\begin{array}{ll}
                I_r & 0\\
                0   & 0
                \end{array}
 \Big),
$
\begin{eqnarray}
&& QBP=\left (
\begin{array}{ll}
                B_{11} & B_{12}\\
                B_{21}   & B_{22}
                \end{array}
 \right ),\label{QBP}
\end{eqnarray}
then we have
\begin{eqnarray}
\left (
\begin{array}{ll}
                I_{r} & 0\\
                0  & 0
                \end{array}
 \right )QBP=\left (
\begin{array}{ll}
                B_{11} & B_{12}\\
                0  & 0
                \end{array}
 \right ),~~QBP\left (
\begin{array}{ll}
                I_{r} & 0\\
                0  & 0
                \end{array}
 \right )=\left (
\begin{array}{ll}
                B_{11} & 0\\
               B_{21} & 0
                \end{array}
 \right ),\label{QBP-Ir}
\end{eqnarray}
which further means
\begin{eqnarray}
&& \mbox{det}(\lambda I_N-AB)=\lambda^{N-r}\mbox{det}(\lambda
I_r-B_{11})=\mbox{det}(\lambda I_N-BA).\label{equality}
\end{eqnarray}
We complete the proof.
\end{proof}

\begin{lemma}
\label{lem:A-2} Assuming that $B$ is an arbitrary $N$th-order complex matrix and
$\bar{B}$ is its complex conjugate, then $\mbox{det}(\lambda
I-\bar{B}B)$ is a polynomial of $\lambda$ with real coefficients.
\end{lemma}

\begin{proof}
Write
\begin{eqnarray}
f(\lambda)=\mbox{det}(\lambda
I_N-\bar{B}B)=a_N\lambda^N+\cdots+a_1\lambda+a_0.
\label{poly}
\end{eqnarray}
Then using Lemma \ref{lem:A-1} with  $A=\bar{B}$  we have
\begin{eqnarray}
f(\lambda)=\mbox{det}(\lambda I_N-\bar{B}B)=\mbox{det}(\lambda
I_N-B\bar{B})=\overline{\mbox{det}(\bar{\lambda}
I_N-\bar{B}B)}=\overline{f(\bar{\lambda})},\label{poly-real}
\end{eqnarray}
which means all the coefficients $\{a_j\}$ are real.
\end{proof}

\section{Discussions on the trivial solutions to the CES \eqref{mKdV-condition}}

Let the square matrices $\mathbb{A}, \mathbb{B}$ follow the relation
\begin{equation}
\mathbb{A}=\mathbb{B}\bar{\mathbb{B}}.
\label{A-B}
\end{equation}

We start from the $2\times 2$ case. Noting that
the product of all the eigenvalues of $\mathbb{A}$ is non-negative,
in the following we first look at
\begin{equation}
\mathbb{A}=\left ( \begin{array}{ll}
                -k^2_1 & 0\\
                0  & -k^2_2
                \end{array}
        \right ),~~~ k_1\neq k_2\neq 0,~~ k_1,k_2\in \mathbb{R},
\end{equation}
and suppose
\begin{equation}
\mathbb{B}=\left ( \begin{array}{ll}
                a & b\\
                c & d
                \end{array}
        \right )
\end{equation}
with undetermined $a,b,c,d\in \mathbb{C}$.
However, in this case it can be found that the matrix relation \eqref{A-B} does not have any solutions unless $k_1^2=k_2^2$.
So next we turn to consider
\begin{eqnarray}
\mathbb{A}=\left ( \begin{array}{ll}
                -k^2 & 0\\
                0  & -k^2
                \end{array}
        \right ),~~k\in \mathbb{R}.
\label{A-1}
\end{eqnarray}
In this case, 
the equation \eqref{A-B} admits a non-diagonal matrix solution
$\mathbb{B}$ as
\begin{eqnarray}
\mathbb{B}=\left ( \begin{array}{cc}
                -d & -\frac{d^2+k^2}{c}\\
                c & d
                \end{array}
        \right )e^{i\theta}, ~~c \ne 0, ~ ~c,d,\theta\in \mathbb{R}. \label{B-1}
\end{eqnarray}
We note that such a $\mathbb{B}$ does not lead to any nontrivial
solutions to the mKdV equation. In fact, in the CES
\eqref{mKdV-condition}, the general solution to the equation set
\eqref{mKdV-condition-a} and \eqref{mKdV-condition-c} is
\begin{subequations}\label{C-6}
\begin{equation}
\varphi=\left(\begin{array}{c}
\varphi_1\\
\varphi_2
\end{array}
\right)=H \left(\begin{array}{c}
e^{i\xi}\\
e^{-i\xi}
\end{array}
\right)
\end{equation}
with arbitrary matrix $H \in \mathbb{C}_{2\times 2}$ and
\begin{equation}
\xi=kx+4k^3t+\xi^{(0)},~~ k,\xi^{(0)}\in \mathbb{R}.
\end{equation}
\end{subequations}
However, no matter what condition the matrix $H$ should satisfy under the equation \eqref{mKdV-condition-b},
the Wronskian
\[f(\varphi)=|H|f((e^{i\xi},e^{-i\xi})^T)\]
is always a constant,
which leads to a trivial solution to the mKdV equation.
In the case of the $N\times N$ matrix ($N$ is even)
\[A=\mathrm{diag}(-k^2,-k^2,\cdots,-k^2),~~k\in \mathbb{R},\]
similar to \eqref{C-6}, the general solution to \eqref{mKdV-condition-a} and \eqref{mKdV-condition-c} is
\[\phi=H\times (e^{i\xi},e^{-i\xi},0,0,\cdots,0)^T,~~H\in \mathbb{C}_{N\times N},\]
which leads to a zero Wronskian $f(\varphi)$.

With these discussions we can conclude that
for the CES \eqref{mKdV-condition} with an $N\times N$ matrix
\begin{equation}
\mathbb{A}=\mathrm{diag}(\alpha,\alpha,\cdots,\alpha),~~\alpha\in \mathbb{R},
\end{equation}
the possible solution $\varphi$ to \eqref{mKdV-condition}
composes a trivial Wronskian $f(\varphi)$.

\section{Proof for Theorem \ref{Th 4.1}}\label{A:sec-4}
\begin{proof}
The parameter $v_0$ will lead to complicated expressions for $\bar{f}$ and its derivatives.
For simplification let us introduce the notation $|\,\cdot\,|_j$ where the subscript $j$ indicates
the absence of the $\phi^{(j)}$ column\cite{Nimmo-Freeman-JPA}, for example,
\[|\h{N}| _{j}=|\phi ^{(0)},\cdots, \phi ^{(j-1)},\phi ^{(j+1)},\cdots,\phi ^{(N)}|.\]
Derivatives of $f$ have already given in \eqref{fde}.
For $\bar{f}$,  using the CES \eqref{cond} we can reach
\begin{subequations}
\begin{align*}
\bar{f}=&|B^{-1}(t)|\sum\limits_{j = 0}^{N} {( -{v_0}} {)^j}{i^{N-j}}{| {\widehat N }|_j},\\
\bar{f}_{X}=&|B^{-1}(t)|\big(\sum\limits_{j = 0}^{N - 1}{( - {v_0}})^{j+1}{i^{N-j-1}}{|{\widehat N} |_j}
+ \sum\limits_{j = 0}^{N - 1} {( - {v_0}} {)^j}{i^{N - j}}{| {\widehat {N - 1} ,N + 1} |_j}\big),\\
\bar{f}_{XX}=&|B^{-1}(t)|\big(( - {v_0})^N| {\widehat {N - 2} ,N + 1}|+ \sum\limits_{j = 0}^{N - 1} {( - {v_0}}
)^j{i^{N - j}}{| {\widehat {N - 1},N + 2}|_j}   \\
&+ \sum\limits_{j = 0}^{N - 2} {( - {v_0}} {)^j}{i^{N - j}}{|
{\widehat {N - 2} ,N,N + 1}|_j} +  \sum\limits_{j = 0}^{N - 2} {( - {v_0}} {)^{j + 2}}{i^{N - j - 2}}{|{\widehat N}|_j} \nonumber \\
&+2\sum\limits_{j =
0}^{N - 2} {( - {v_0}} {)^{j + 1}}{i^{N - j - 1}}{|{\widehat{N-1} ,N
+ 1}|_j}\big),
\end{align*}
\begin{align*}
\bar{f}_{XXX}=&|B^{-1}(t)|\big(3\sum\limits_{j = 0}^{N -
3} {( - {v_0}})^{j + 2}{i^{N - j - 2}}{| {\widehat {N - 1} ,N +
1} |_j}
+ ( - {v_0})^N| {\widehat {N - 2} ,N + 2}| \nonumber \\
&+3\sum\limits_{j = 0}^{N - 2} {( - {v_0}} {)^{j + 1}}{i^{N - j - 1}}{| {\widehat {N - 1} ,N + 2} |_j}
+2{( - {v_0})^N}| {\widehat {N - 3} ,N - 1,N + 1} | \nonumber \\
& +\sum\limits_{j = 0}^{N - 1} {( - {v_0} )^j}{i^{N - j}}{| {\widehat {N - 1} ,N + 3} |_j}
+2\sum\limits_{j = 0}^{N - 2} {( - {v_0}} {)^j}{i^{N - j}}{| {\widehat {N - 2} ,N,N + 2} |_j} \nonumber \\
&+3\sum\limits_{j = 0}^{N - 3} {( - {v_0}} {)^{j + 1}}{i^{N -
j-1}}{| {\widehat {N - 2} ,N,N + 1} |_j} +2{( - {v_0})^{N - 1}}i| {\widehat {N - 3} ,N,N + 1} | \nonumber \\
&+\sum\limits_{j = 0}^{N - 3} {( - {v_0}} {)^j}{i^{N - j}}{| {\widehat {N - 3} ,N - 1,N,N + 1} |_j}
+\sum\limits_{j = 0}^{N - 3} {( - {v_0}} {)^{j + 3}}{i^{N - j - 3}}{| {\widehat N } |_j} \big),
\end{align*}
\begin{align*}
\bar{f}_t=&-4|B^{-1}(t)|{\big(\sum\limits_{j = 0}^{N -
1}{(-{v_0}})^j}{i^{N - j}}{| {\widehat {N - 1} ,N + 3} |_j}
+\sum\limits_{j = 0}^{N - 3} {( - {v_0}} {)^{j + 3}}{i^{N - j - 3}}{|
{\widehat N } |_j} \nonumber\\
&+ \sum\limits_{j = 0}^{N - 3} {( - {v_0}} {)^j}{i^{N - j}}{| {\widehat {N - 3} ,N - 1,N,N + 1} |_j}
+( - {v_0})^N| {\widehat {N - 2} ,N + 2} | \\
& - \sum\limits_{j = 0}^{N - 2} {( - {v_0}} {)^j}{i^{N - j}}{| {\widehat {N - 2} ,N,N + 2} |_j}
  - {( - {v_0})^{N -
1}}i| {\widehat {N - 3} ,N,N + 1} | \nonumber \\
  & - {( - {v_0})^N}| {\widehat {N - 3} ,N - 1,N + 1} |\big)+
\mathrm{tr}({C}(t))|{{B^{-1}}(t)}|\sum\limits_{j = 0}^N {( - {v_0}}{)^j}{i^{N - j}}{| {\widehat N } |_j}.
\end{align*}
\end{subequations}
Using the condition \eqref{cond} the complex conjugate of \eqref{f-t} should be same as the above $\bar{f}_t$.
This requires $\mathrm{tr}(C(t))\in \mathbb{R}(t)$, i.e., the condition \eqref{tr-Ct}.

Besides, noting that $\phi_{XX}=(B(t)\bar{B}(t)-v_{0}^2I_N)\phi$ and
using Proposition \ref{Prop 2.1} with $\Omega_{j,s}=\partial_X^2$,
we can have the following identities:
\begin{align*}
\mbox{tr}(B(t)\bar{B}(t)-v_0^2I_N){| {\widehat N } |_j}=&{| {\widehat {N - 1} ,N + 2} |_j} - {| {\widehat {N - 2} ,N,N + 1} |_j}\\
&- {| {\widehat N } |_{j - 2}},~~~~~~~~~~~~~~~~(j = 0,1, \cdots ,N - 2), \\
\mbox{tr}(B(t)\bar{B}(t)-v_0^2I_N){| {\widehat {N - 1} ,N + 1} |_j} = &{| {\widehat {N - 1} ,N + 3} |_j} - {| {\widehat {N - 3} ,N - 1,N,N + 1} |_j} \nonumber\\
&- {| {\widehat {N - 1} ,N + 1} |_{j - 2}},~~(j = 0,1, \cdots ,N - 3), \\
\mbox{tr}(B(t)\bar{B}(t)-v_0^2I_N)| {\widehat {N - 1} } | =  &- | {\widehat {N - 3} ,N - 1,N} | + | {\widehat {N - 2} ,N + 1} |, \\
\mbox{tr}(B(t)\bar{B}(t)-v_0^2I_N)| {\widehat {N - 2} ,N} | = & - |
{\widehat {N - 4} ,N - 2,N - 1,N} | + | {\widehat {N - 2} ,N + 2} |.
\end{align*}
With these results and \eqref{fde} in hand, for \eqref{eq4a} we have
\begin{align*}
&\bar{f_{t}}f-\bar{f}f_t+\bar{f}_{XXX}f-3\bar{f}_{XX}f_X+3\bar{f_X}f_{XX}-\bar{f}f_{XXX}\\
=&6{( - {v_0})^{N - 1}}i(| {\widehat {N - 3} ,N,N + 1} || {\widehat
{N - 1} } | - | {\widehat {N - 2} ,N} || {\widehat {N - 3} ,N - 1,N
+ 1} |\\
&+ | {\widehat {N - 2} ,N + 1} || {\widehat {N - 3} ,N - 1,N} |)\label{v0-WP-2} \\
=& 0 .
\end{align*}
Similarly, one can prove \eqref{eq4b}.
\end{proof}

\end{document}